\documentclass[hidelinks,conference]{IEEEtran}
\IEEEoverridecommandlockouts
\IEEEpubid{\makebox[\columnwidth]{978-1-6654-4895-6/21/\$31.00~
\copyright2021 IEEE \hfill} 
\hspace{\columnsep}\makebox[\columnwidth]{ }}

\usepackage{colortbl}
\usepackage{amsmath,amssymb,amsthm}
\usepackage{mathtools}
\usepackage{hyperref}
\usepackage{cleveref}
\usepackage{subcaption}
\usepackage{url}
\usepackage{accents}
\usepackage{bm}
\usepackage{booktabs}
\usepackage{tikz}
\usetikzlibrary{decorations.pathreplacing,calligraphy}

\usepackage{xspace}
\makeatletter
\def\namedlabel#1#2{\begingroup
   \def\@currentlabel{#2}%
   \label{#1}\endgroup
\makeatother}

\newtheorem{theorem}{Theorem}
\newtheorem{lemma}[theorem]{Lemma}
\newtheorem{corollary}[theorem]{Corollary}
\newtheorem{claim}[theorem]{Claim}
\newtheorem{example}[theorem]{Example}
\newtheorem{definition}[theorem]{Definition}

\crefname{theorem}{Theorem}{Theorems}
\crefname{lemma}{Lemma}{Lemmas}
\crefname{corollary}{Corollary}{Corollaries}
\crefname{claim}{Claim}{Claims}
\crefname{definition}{Definition}{Definitions}
\crefname{example}{Example}{Examples}
\crefname{figure}{Figure}{Figures}

\newcolumntype{C}[1]{>{\centering\arraybackslash}m{#1}}
\newcolumntype{L}[1]{>{\arraybackslash}m{#1}}

\DeclareRobustCommand{\gen}{\text{\reflectbox{$\neg$}}}
\newcommand{\pre}[1]{\protect\accentset{\smash{\raisebox{-0.14ex}{\scalebox{.9}{\hspace{.04em}$\mathrlap{\scriptscriptstyle\bm\gen}\hspace{-0.02em}\scriptscriptstyle\bm\gen$}}}}{#1}}
\newcommand{\suf}[1]{\protect\accentset{\smash{\raisebox{-0.14ex}{\scalebox{.9}{$\mathrlap{\scriptscriptstyle\bm\neg}\hspace{0.02em}\scriptscriptstyle\bm\neg$\hspace{.04em}}}}}{#1}}
\renewcommand{\inf}[1]{\protect\accentset{\smash{\raisebox{0.155ex}{\scalebox{1}{\rule{.36em}{.4pt}\hspace{.00em}}}}}{#1}}
\newcommand{\loo}[1]{\protect\accentset{\smash{\raisebox{-0.10ex}{\scalebox{.6}{$\scriptscriptstyle\bm\circ$}}}}{#1}}

\newcommand{\tr}[3][1.5em]{\xrightarrow[{\makebox[#1]{\raisebox{(\heightof{$m$}-\heightof{$#3$})/2+.4ex}[0ex][0ex]{$#3$}}}]{\makebox[#1]{\raisebox{(\heightof{$#2$}-\heightof{$m$}*98/99)*1/5-0.25ex}[0ex][0ex]{$#2$}}}}
\newcommand{\N}{\mathbb{N}}
\newcommand{\nfa}{\ensuremath{\textrm{NFA}}\xspace} %
\newcommand{\nft}{\ensuremath{\textrm{NFT}}\xspace} %
\newcommand{\dft}{\ensuremath{\textrm{DFT}}\xspace} %
\newcommand{\dfa}{\ensuremath{\textrm{DFA}}\xspace} %
\newcommand{\tdfa}{\ensuremath{\textrm{2t-DFA}}\xspace} %
\newcommand{\ftdfa}{\ensuremath{\textrm{fv-2t-DFA}}\xspace} %
\newcommand{\fnft}{\ensuremath{\textrm{fv-NFT}}\xspace} %
\newcommand{\lnft}{\ensuremath{\lambda\textrm{-NFT}}\xspace} %
\newcommand{\funtdfa}{\ensuremath{\textrm{f-2t-DFA}}\xspace} %
\newcommand{\funnft}{\ensuremath{\textrm{f-NFT}}\xspace} %
\newcommand{\blank}{\textnormal{\texttt{\char32}}}
\newcommand{\shortcut}{\ensuremath{g}\xspace}
\newcommand{\outputspeed}{\ensuremath{s}\xspace}
\newcommand{\traillength}{\ensuremath{t}\xspace}
\newcommand{\homestretch}{\ensuremath{h}\xspace}
\newcommand{\pairset}{\ensuremath{P}\xspace}
\newcommand{\buffersize}{\ensuremath{r}\xspace}
\newcommand{\move}{\ensuremath{m}\xspace}

\newcommand{\caseA}{\ensuremath{{\begin{tikzpicture}[scale=1.5,
line width=.4pt]\draw (0,.16)--+(.08,-.16)--+(-.08,-.16)--cycle;
\node () at (.03,.075) {};\end{tikzpicture}}}}
\newcommand{\boldcaseA}{\ensuremath{{\begin{tikzpicture}[scale=1.5,
line width=.8pt]\draw (0,.16)--+(.08,-.16)--+(-.08,-.16)--cycle;
\node () at (.03,.075) {};\end{tikzpicture}}}}
\newcommand{\indexcaseA}{\ensuremath{{\begin{tikzpicture}[scale=0.8,
line width=.04pt]\draw (0,.16)--+(.08,-.16)--+(-.08,-.16)--cycle; 
\end{tikzpicture}}}}
\newcommand{\caseB}{\ensuremath{{\begin{tikzpicture}[scale=1.5,
line width=.4pt]\draw (0,0)--+(.08,.16)--+(-.08,.16)--cycle; 
\node () at (.03,.075) {};\end{tikzpicture}}}}
\newcommand{\boldcaseB}{\ensuremath{{\begin{tikzpicture}[scale=1.5,
line width=.8pt]\draw (0,0)--+(.08,.16)--+(-.08,.16)--cycle; 
\node () at (.03,.075) {};\end{tikzpicture}}}}
\newcommand{\indexcaseB}{\ensuremath{{\begin{tikzpicture}[scale=0.8,
line width=.04pt]\draw (0,0)--+(.08,.16)--+(-.08,.16)--cycle; 
\end{tikzpicture}}}}

\title{From Finite-Valued Nondeterministic Transducers to 
Deterministic Two-Tape Automata} 

\author{\IEEEauthorblockN{Elisabet Burjons}
\IEEEauthorblockA{Department of Computer Science\\
RWTH Aachen, Germany\\
burjons@cs.rwth-aachen.de}
\and
\IEEEauthorblockN{Fabian Frei}
\IEEEauthorblockA{Department of Computer Science\\
ETH Zürich, Switzerland\\
fabian.frei@inf.ethz.ch}
\and
\IEEEauthorblockN{Martin Raszyk}
\IEEEauthorblockA{Department of Computer Science\\
ETH Zürich, Switzerland\\
martin.raszyk@inf.ethz.ch}}

\begin{document}
\maketitle

\begin{abstract}
The question whether P equals NP revolves around the discrepancy 
between active production
and mere verification by Turing machines.
In this paper, we examine the analogous problem for 
finite transducers and automata.
Every nondeterministic finite transducer defines a binary relation
associating each input word with all output words that 
the transducer can successfully produce on the given input.
Finite-valued transducers are those for which there is 
a finite upper bound on the number 
of output words that the relation associates with every input word.
We characterize 
finite-valued, functional, and unambiguous 
nondeterministic transducers 
whose relations can be verified by 
a deterministic two-tape automaton,
show how to construct such an automaton if one exists,
and prove the undecidability of the criterion. 
\end{abstract} 

\medskip
\section*{Formal Verification}
All major results presented in this paper were also 
formally proved in Isabelle/HOL~\cite{NPW02}, a theorem prover
using higher-order logic,
ensuring their correctness beyond any reasonable doubt. The 
formalizations, including 
an explanation of how to interpret them and details on how they 
connect to the individual definitions and theorems in the 
present paper, are available at 
\url{https://github.com/lics21-automata/automata}.
\medskip
\medskip
\section{Introduction}
\indent\textbf{Automata.}
One of the simplest computation models is a finite automaton 
reading any given input word from left to right while 
deterministically updating its state according to a transition 
function that depends on the current state and symbol that is 
currently read; once the input word has been read completely, it is 
either accepted or rejected depending on the state in which the 
automaton ends up. 
We refer to such an automaton as 
a deterministic finite automaton (\dfa).
For a nondeterministic finite automaton (\nfa), in contrast, the 
transition function may permit an arbitrary number of transitions  
from the current state instead of exactly one; the input word is 
then accepted if it is accepted for any of the possible choices of 
transitions. 
The set of words accepted by an automaton $A$ is 
its recognized language $L(A)$. 
It is well-known that deterministic and 
nondeterministic finite automata both recognize the same class
of \emph{regular} languages~\cite{RS59}; that is, 
nondeterminism does not add any computational power
to the finite automaton model.

\medskip
\textbf{Two-Tape Automata.}
One way to generalize finite automata is the multi-tape model.
Here, an automaton has multiple tapes, each with one head that 
can read the word written on the tape from left to right.
There are several different but equivalent ways to model 
multi-tape automata, most notably the Rabin-Scott model and 
the Turing machine model.
In the Rabin-Scott model, only one of the heads is reading 
a symbol and advancing to the next one in each step; 
the current state determines on which tape this happens. 
The computation is either accepted or rejected once all heads 
have reached the end of the tape. 
In the Turing machine model, in contrast, the multi-tape automaton
reads all symbols at the current head positions simultaneously a
nd can then move forward any subset of the heads in each step.
In this paper, we only examine deterministic finite automata with 
two tapes (\tdfa{}s).

\medskip
\textbf{Transducers.}
\emph{Transducers} were introduced by Elgot and Mezei~\cite{EM65} 
as a generalization of automata. 
A transducer reads the given input word symbol by symbol, 
just like an automaton, but additionally produces 
a separate output word while doing so. 
During each transition reading one input symbol, the transducer 
produces a possibly empty sequence of output symbols. 
The final output word is obtained by concatenating 
the outputs of all transitions. 
After transducing the entire input word into an output word, the 
transducer decides whether the computation is accepting or not 
depending on the current state. 
In this paper, we are mainly interested in nondeterministic finite 
transducers 
(\nft{}s). For further context, we will also 
consider deterministic finite transducers (\dft{}s) and an 
extension of \nft{}s that allows for $\lambda$-transitions 
(\lnft{}s). Allowing $\lambda$-transitions means that the 
transducer can produce output not only when reading 
an input symbol but also on the empty word $\lambda$, that is, 
at any time without reading any part of the input.

\medskip
\textbf{Functionality and Finite-Valuedness.}
A deterministic transducer associates any given input word with the 
output word produced on it if the corresponding computation is 
accepting. 
A nondeterministic transducer can have multiple accepting 
computations on an input word, associating it with all 
corresponding output words. 
Denoting the input and output alphabet by $\Sigma$ and $\Gamma$, 
respectively, a transducer's language is a relation 
between $\Sigma^\ast$ and $\Gamma^\ast$. 
For two-tape automata, we analogously consider the tapes 
as containing an input over $\Sigma$ and an output over $\Gamma$, 
respectively. Thus, the language of a two-tape automaton is 
a relation $L\subseteq \Sigma^\ast\times \Gamma^\ast$ as well.

Given such a relation, 
we say that an input $a\in \Sigma^\ast$ is \emph{associated} 
with an output $u\in\Gamma^\ast$ if $(a,u)\in L$. 
The relation $L$ is \emph{functional} if any input is 
associated with at most one output. 
Note that such an $L$ describes a function that 
may be only partial. 
Given a $k\in\N$, we say that $L$ is \emph{$k$-valued} 
if it associates at most $k$ outputs 
with every input. 
If there is a $k\in\N$ such that $L$ is $k$-valued, 
we call $L$ finite-valued. 
We call two-tape automata and transducers functional 
or finite-valued if their languages are. 
We denote the restriction of machine model M to functional 
or finite-valued ones by f-M and fv-M, respectively. 
Moreover, we denote by $\mathcal{L}(\textnormal{M})$ 
the family of languages recognizable by machines of model M. 

See \cref{fig:modelcompare} for an overview of the 
hierarchy of languages that are of interest to us. 
The vertical inclusions in it are trivial since 
functionality implies finite-valuedness. 
The strictness follows for instance from the relations  
$\{(0,0),(0,00),\dots,(0,0^k)\}$---which is $k$-valued but not 
functional for any $k>1$---and the not even 
finite-valued $\{\,(0,0^i)\mid 
i\in\N\,\}$ since they are 
recognizable by \tdfa{}s as well as \nft{}s. 

It is worth pointing out that determinism implies functionality for 
transducers but not for two-tape automata.

\medskip
\textbf{Motivation.}
The main motivation underlying all computation models is to study 
their expressive power, that is, the set of languages recognized 
by machines in each model. 
For instance, \dft{}s are strictly less expressive than 
functional \nft{}s. 
A characterization of the languages that can be 
recognized in one given model but not the other is also of 
interest. 

For example, a classical result by Choffrut shows that satisfying 
the so-called \emph{twinning property (condition de jumelage)} is 
a both necessary and sufficient condition 
for a functional \nft{} to be determinizable~\cite{CH77}. 
Choffrut also introduced the notion of functions \emph{with bounded 
variation (à variation bornée)}---called \emph{uniformly bounded} 
by others~\cite{RS91}---which characterizes the functions 
recognized 
by determinizable \nft{}s.
Later, this criterion of bounded variation was proved to be 
decidable in polynomial 
time~\cite{WK95}. 

Our goal is to develop an analogous property that characterizes 
functional \nft{}s having an equivalent \tdfa, to examine the 
decidability or undecidability of this criterion, and provide an 
algorithm that transforms an \nft{}s into an equivalent \tdfa 
whenever possible.

\begin{figure}
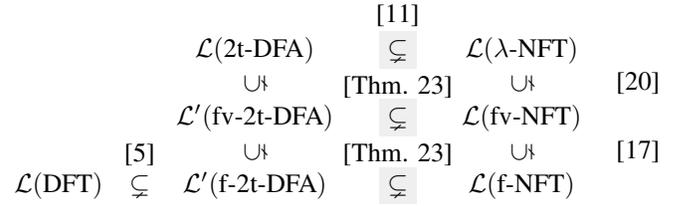

\setlength{\tabcolsep}{4pt}
\centering
\begin{tabular}{cccccc}
&&&\cite{FR68}&&\\
&&$\mathcal{L}(\tdfa)$&\colorbox{lightgray!25}{$\subsetneq$}&$\mathcal{L}(\lambda\textrm{-}\nft)$&\\
&&\rotatebox[origin=c]{90}{$\subsetneq$}&{\hspace*{-1ex}\raisebox{-0.3ex}{[Thm.~\ref{thm:undecidable}]}\hspace*{-1ex}}&\rotatebox[origin=c]{90}{$\subsetneq$}&\hspace*{.5em}\cite{W89}\\
&&$\mathcal{L}'(\ftdfa)$&\colorbox{lightgray!25}{$\subsetneq$}&$\mathcal{L}(\fnft)$&\\
&\raisebox{-0.3ex}{\cite{CH77}}&\rotatebox[origin=c]{90}{$\subsetneq$}&{\hspace*{-1ex}\raisebox{-0.3ex}{[Thm.~\ref{thm:undecidable}]}\hspace*{-1ex}}&\rotatebox[origin=c]{90}{$\subsetneq$}&\hspace*{.5em}\cite{S75}\\
$\mathcal{L}(\dft)$&$\subsetneq$&$\mathcal{L}'(\funtdfa)$&\colorbox{lightgray!25}{$\subsetneq$}&$\mathcal{L}(\funnft)$&
\end{tabular}
 \caption{The strict hierarchy of languages recognized by relevant 
 machine models. The inclusions and their strictness are mostly
 trivial; they are briefly discussed in the introduction.
 For each inclusion, there is the following decision problem: 
 Given an element of the superset, is it contained in the subset?
 The shaded inclusions are exactly those associated 
 with an undecidable problem. 
 The references point to the proofs of decidability 
 and undecidability, respectively. 
 The decidability for the two vertical inclusions is an immediate 
 consequence of the same result for the two right 
 ones.
}
 \label{fig:modelcompare}
\end{figure} 

\medskip
\textbf{Inclusion Hierarchy.} To provide the context for our main 
endeavor, we prove all inclusions in \cref{fig:modelcompare}, going 
from left to right. 
The inclusion $\mathcal{L}(\dft)\subseteq\mathcal{L}(\tdfa)$ has 
been formally proved by Aho et al.~\cite[Thms.~3.1 and~3.2]{AHU69}. 
The intuition is simple: 
Given a \dft{}, we can construct an equivalent \tdfa{} that 
simply simulates the \dft{} on the first tape
and checks whether the output of the \dft{} matches 
the content of the second tape.
It is easy to construct an example, such as the following, showing 
that the inclusion is strict.

\begin{example}
We present the concrete language that can be recognized by an 
\funtdfa but not by a \dft. 

Consider the relation that maps any binary string $w$ to $0^{|w|}$ 
if the last symbol of $w$ is $0$ and to $1^{|w|}$ otherwise. 
It is easy to construct a functional \tdfa that verifies this 
relation as follows. 

The automaton advances its two heads synchronously, always 
comparing the current output symbol with the previous one. If there 
is a discrepancy at any point, the automaton rejects; otherwise, 
the output word is guaranteed to be uniform. If the two heads reach 
the end of the input and output word at the same time and the last 
two symbols match, then the automaton accepts; otherwise, it 
rejects. 

A \dft cannot compute this relation, however, 
because it cannot output any symbol until it knows the last symbol 
of the word on the first tape,
and by that time, the transducer cannot recall the length of the 
input word, which equals the number of copies of the 
last symbol that the transducer needs to produce.
\end{example}  

To prove the inclusion 
$\mathcal{L}(\tdfa)\subseteq \mathcal{L}(\lnft)$, 
we can construct a transducer that guesses the output and then 
verifies it by simulating the automaton. 
Here, we need to consider transducers with 
$\lambda$-transitions---i.e., transducers that 
can produce output without reading any input symbol---for 
the inclusion to hold because a 
general two-tape automaton can advance its output head 
arbitrarily many times while its input head is standing still.

Trying to prove the two inclusions 
$\mathcal{L}(\ftdfa)\subseteq \mathcal{L}(\fnft)$ 
and  
$\mathcal{L}(\funtdfa)\subseteq 
\mathcal{L}(\funnft)$, 
we can essentially let the transducer guess the 
output and then verify it by simulating the 
automaton.
This works because the output head of a 
finite-valued 
or functional 
automaton can advance by at most $n$ positions,
where $n$ denotes the number of the automaton's
states, while the input head is standing still. 
Otherwise, one could pump in a non-empty
output subword on which the automaton loops 
without changing the input word.
However, one needs to exclude computations on the empty input word 
$\lambda$ since the transducer cannot produce (and 
thus guess) any 
output without reading an input symbol.
For this technical reason, we introduce the 
following notation used in \cref{fig:modelcompare}. 
We denote by  
$\mathcal{L}'(\ftdfa)=\{L\smallsetminus(\{\lambda\}\times\Gamma^\ast)\mid
L\in \mathcal{L}\}$ the set of 
\ftdfa{}-recognizable 
relations with all pairs $(\lambda, x)$ 
associating the empty 
input $\lambda$ to any output removed, and 
$\mathcal{L}'(\funtdfa)$ is 
defined analogously.

\textbf{Question of Decidability.}
We now turn our attention to the corresponding 
decidability questions.  
Whether a machine is finite-valued is decidable even 
in the case of nondeterministic transducers~\cite{W89}.
We can also decide whether the given \fnft is functional 
or not~\cite{S75}. 
It follows from these two results that we can decide 
the finite-valuedness and functionality for two-tape automata 
as well. 

Via a reduction from the Post correspondence problem, 
Fischer and Rosenberg proved 
that, for any $k\ge2$, it is undecidable whether 
a $k$-tape \nfa has any equivalent $k$-tape \dfa~\cite{FR68}.
In particular, given a $\textrm{2t-}$\nfa, it is 
undecidable whether there is an equivalent \tdfa. 
This can be reformulated with \lnft{}s, as
they are equivalent to $\textrm{2t-}$\nfa{}s. %
Their proof relies on relations that are not finite-valued, however.
In this paper, we extend Fischer and Rosenberg's 
undecidability result~\cite{FR68} to the finite-valued, 
functional, and even unambiguous case. 
The undecidability also implies the strictness 
of the corresponding inclusions. 

For the remaining decision problem, we can construct 
from the given \funtdfa an equivalent \nft, 
as long as we consider automata accepting relations 
with non-empty input words, as described above.
Finally, we can check whether the described \nft, which is 
equivalent to the given \ftdfa, has an equivalent \dft via the 
decidable twinning property~\cite{CH79}.\medskip

\medskip
\textbf{Alternative Terminology.}
We mention some alternative terminology that will not be relevant 
in presentation of our paper but can be found in the literature. 
Deterministic transducers are sometimes called \emph{sequential} 
transducers. This is because their transductions---which are unique 
for every given input word---can be processed sequentially, symbol 
by symbol. 

There is also the notion of \emph{subsequential} transducers. 
These are sequential transducers with one additional feature: 
After reading the entire input word and having reached a final 
state, they can append to the output one more word that may depend 
on the final state. 

The difference between sequentiality and subsequentiality is 
an overall rather insubstantial technical detail. 
It is easy to see that sequential and subsequential
become 
equivalent under the assumption of a unique symbol marking the end 
of every input word. 

A function is called sequential or subsequential if it is 
recognized as the relation of a sequential or subsequential 
transducer, respectively. 
Finally, a relation (and in particular a function) is called 
\emph{rational} if it is recognized by a finite nondeterministic 
transducer.

We can now restate our motivating result by Choffrut in his own 
terminology: 
Rational functions with bounded variation are exactly the 
subsequential functions. 

It is also possible to extend the notion of subsequentiality to 
nondeterministic transducers. 
However, the additional option to add a suffix to the output 
depending on the final state makes even less of a difference if 
nondeterministic transitions are allowed. This is because the 
transducer can always guess that the current symbol is the last 
with an additional transition option including the final suffix 
for every situation. Whenever 
the guess turns out to be wrong, the corresponding computation path 
can be 
aborted in the following step. 
This works on all input words except for the empty one: 
On the empty input word $\lambda$ a transducer cannot perform even 
a single  
transition and is thus unable to produce any output at all without 
subsequentiality. 
However, even with subsequentiality an \nft can only associate 
finitely many 
output words with the input $\lambda$; the difference is thus 
negligible and we may just ignore it by excluding all pairs 
$\{\,(\lambda,u)\mid u\in\Gamma^\ast\,\}$ in the relation from 
consideration.

\medskip
\textbf{Further Model Extensions.}
We remark that both multi-tape automata and transducers can be 
extended to the \emph{two-way} model that allows the heads to move 
not only forward but also back, from right to 
left~\cite{C12,CES17}. 
Filiot et al.~\cite{FGRS13} showed that it is decidable 
whether a two-way \dft{} can be transformed into 
an equivalent one-way \dft and provide a characterization of 
the two-way \dft{}s for which this is possible. 
It is known that two-way \dft{}s compute exactly 
those relations that are expressible as
monadic second-order logic string transductions~\cite{EH99}.
Rabin and Scott~\cite{RS59} proved that two-way automata are 
as powerful as one-way automata for a single tape 
but recognize strictly more languages with multiple tapes.

Yao and Rivest have shown that increasing the number of heads on a 
single tape yields a strict hierarchy of languages~\cite{YR78}; see 
also the survey by Holzer et al.~\cite{HKM11}.
A more recent result by Raszyk et al.~\cite{RBT19} proves that 
\dft{}s with multiple heads are 
strictly more expressive than \funnft{}s. 
Asking whether multi-head two-way \dfa{}s can simulate two-head 
one-way \nfa{}s is equivalent to the famous question whether 
$\textrm{L}=\textrm{NL}$~\cite{Su75}.

In the present paper, we exclusively consider the one-way model 
with 
one head per tape.

\medskip
\textbf{Our Contributions.}
We compare the computational power of \tdfa{}s and \fnft{}s. 
To this end, we introduce the notion of transducers having 
\emph{bounded trailing} and show that this property characterizes 
\fnft{}s whose relations are computable by a \tdfa. In other 
words: Any given \fnft{} (and thus in particular any functional 
\nft) 
has an equivalent \tdfa if and only if the \fnft{} 
has bounded trailing. 

This result is analogous to the already mentioned classical result 
by 
Choffrut~\cite{CH79}. He introduced the notion of a transducer 
satisfying the \emph{twinning property} and 
proved that this property characterizes functional \nft{}s
having an equivalent \dft{}.

We show how our notion of bounded trailing can be seen as a 
relaxation of another property due to 
Choffrut~\cite{CH79}, namely functions \emph{with bounded 
variation}. 

Furthermore, we prove the question whether a given transducer has 
bounded trailing to be undecidable.
This stands in contrast to the twinning property, which was proven 
to be a 
polynomial-time decidable criterion by Weber and Klemm~\cite{WK95}. 

Finally, we provide a concrete construction that is guaranteed to 
transform any 
given \fnft into an equivalent \tdfa whenever possible, that is, if 
the \fnft has bounded trailing, despite this criterion being 
undecidable. 

We have formally proved all major results of this paper 
(\cref{lem:bvtt}, \cref{thm:bt-sufficient}, \cref{clm:red-func}, 
\cref{clm:bt-undec}, and \cref{thm:bt-necessary} for functional 
\nft{}s)
as well as \cref{ex:constr,ex:bt} in the proof
assistant Isabelle~\cite{formalization} to ensure their 
correctness. 
We also provide pen-and-paper proofs for all of our results.

\section{Preliminaries}\label{sec:prel}
We formally describe the automaton model and 
transducer model considered in this paper 
and then introduce some useful terminology.

There are many equivalent ways of extending 
the one-tape automaton model to multiple tapes. 
We base our definition on the Turing machine model. 

\begin{definition}\label{def:tdfa}%
A \emph{two-tape deterministic finite automaton} (\tdfa), 
is a septuple $A=(Q,\Sigma,\Gamma,\blank,\delta,q_0,F)$ 
consisting of
\begin{itemize}
\item a finite, nonempty \emph{set of states} $Q$, 
\item a finite, nonempty \emph{input alphabet} $\Sigma$, 
\item a finite, nonempty \emph{output alphabet} $\Gamma$, 
\item a \emph{blank symbol} $\blank\notin \Sigma\cup\Gamma$,
\item an \emph{initial state} $q_0\in Q$,
\item a \emph{set of accepting states} $F\subseteq Q$, and  
\item a \emph{transition function} $\delta$ satisfying the 
following properties: 
\begin{itemize}
\item $
\begin{array}[t]{l}
\delta\colon\ Q\times
(\Sigma\cup\{\blank\})\times(\Gamma\cup\{\blank\})\\
\hspace*{1.5em}\to\hspace*{-0em}
Q\times\{\textnormal{Stay},\textnormal{Advance}\}
\times\{\textnormal{Stay},\textnormal{Advance}\},\\
\hspace*{1.5em}
(q,\sigma,\gamma)
\hspace*{0em}\mapsto\hspace*{0em}
(q',\move_1,\move_2),
\end{array}$
\item $\sigma=\blank\implies\move_1=\textnormal{Stay}$\\(i.e., the 
first head stops at the end delimiter $\blank$), 
\item $\gamma=\blank\implies\move_2=\textnormal{Stay}$ (i.e., the 
second
head stops at the end delimiter $\blank$), and 
\item 
$\move_1=\textnormal{Advance}\;\lor\;\move_2=\textnormal{Advance}$ 
(i.e., at least one head advances in every step).
\end{itemize}
\end{itemize}
\end{definition}

The computation of a \tdfa{} proceeds as follows.
There are two tapes that we call the \emph{input} and 
\emph{output} tape;  
the former contains a word $a\in\Sigma^\ast$, 
the latter a word $u\in \Gamma^\ast$. 
Both words are delimited by the blank symbol $\blank$ at the end. 
The two tapes have one reading head each, 
which we refer to as the \emph{input} and \emph{output} head. 
The input and output head are initially positioned 
on the first symbol of $a$ and $u$, respectively. 
Depending on the current state of the automaton and 
what symbols the two heads are reading,
either the input or output head or both advance 
by one symbol in each step.
As soon as a head reaches the blank symbol delimiting a word, 
it cannot move any further.
The computation ends when both heads have reached 
the end of the input.
A word pair $(a,u)$ is in the language accepted 
by the automaton if and only if the computation on 
this pair of words ends in an accepting state. 
A \emph{configuration} consists of the positions of 
the two heads and the current state. 
We write $q\tr {\raisebox{-.2ex}{$a$}}uq'$ if the automaton 
can start in a state $q$ and end up in a state $q'$
by reading a word $a$ with the input head and $u$ with 
the output head.

We now formally define a nondeterministic finite transducer.

\begin{definition}%
A \emph{nondeterministic finite transducer} (\nft) is a sextuple $T=(Q,\Sigma,\Gamma,\delta,\allowbreak q_0,F)$, where $Q$, $\Sigma$, $\Gamma$, $q_0$, and $F$ are defined as for a \tdfa in \cref{def:tdfa},
but the \emph{transition function} $\delta$
now is a function that maps each pair $(q,\sigma)\in Q\times \Sigma$ of a state and input symbol to a \emph{finite subset} of $Q\times \Gamma^\ast$, describing the nondeterministic transition options.
\end{definition}

The computation of an \nft proceeds like the computation of 
an \nfa except that the \nft produces in each step a 
possibly empty sequence of output symbols. 
The outputs from the single steps are concatenated to obtain 
the final output of the complete computation. 
Using the visualization of a two-tape machine, we can describe 
the computation of an \nft
as reading the word on the input tape symbol-wise 
while writing to the initially empty output tape, 
appending each step's output. 
We write 
$q\tr[.8em] {\raisebox{-.2ex}{$a$}}u q'$ if the transducer 
transitions from state $q$ to state $q'$ with the input head 
reading $a\in\Sigma^\ast$ and the output head 
producing $u\in\Gamma^\ast$. 

The computation ends once the entire input word has been read, 
that is, when the input head has reached the blank symbol. 
If the transducer is in an accepting state at this moment, 
then the word pair $(a,u)$
on the two tapes is in the relation $L(T)$ computed by $T$.
If the transition function does not offer any option for a step and 
thus forcibly ends the computation before the blank symbol on the 
input tape is reached, then this computation does not contribute to 
the relation $L(T)$. 

A binary relation $R\subseteq \Sigma^\ast\times\Gamma^\ast$ is 
\emph{finite-valued} if a constant bounds the number of outputs 
$u\in\Gamma^\ast$ per input $a\in\Sigma^\ast$, that is, if 
$\exists\, k\in\N\colon\,\forall\, a\in\Sigma^\ast\colon\quad 
|\{\,u\in\Gamma^\ast\mid (a,u)\in R\,\}|\le k$.
\begin{definition}%
An \nft $T$ is a \emph{finite-valued} nondeterministic finite 
transducer ($\fnft$) if $L(T)$ is a finite-valued relation. 
If $L(T)$ is even functional, then $T$ is a \emph{functional} 
nondeterministic finite transducer ($\funnft$). 
\end{definition}

We will use the following two notions related to \nft{}s 
many times in our proofs.
\begin{definition}[Shortcut guarantee \shortcut]\label{def:sg}
Let an \nft $T=(Q,\Sigma,\Gamma,\delta,q_0,F)$ be given. Let 
$Q'\subseteq Q$ denote the set of all 
\emph{co-reachable} states, that 
is, states from which an accepting state can be reached on some 
input. For every $q\in Q'$, let 
\[\shortcut_q= \min\{\,|x|\mid \exists\, u\in\Gamma^\ast, f\in 
F\colon 
q\tr[.8em] xuf\,\}\] denote the length of a shortest word $x$ that 
leads from $q$ into an accepting state.
We call $\shortcut(T)=\max_{q\in Q'}\shortcut_q$ the 
\emph{shortcut guarantee} of $T$. 
\end{definition}

\begin{definition}[Output speed \outputspeed]\label{def:os}
Let an \nft $T=(Q,\Sigma,\Gamma,\delta,q_0,F)$ be given. We call 
\[\outputspeed(T)=\max\{\,|\gamma|\mid \exists\, q,q'\in Q, 
\sigma\in\Sigma\colon (q',\gamma)\in \delta(q,\sigma)\,\}\] 
the \emph{output speed} of $T$. Note that $\outputspeed(T)$ is 
well-defined since $\delta$ only maps to finite subsets of $Q\times 
\Gamma^\ast$. 
\end{definition}

\section{Bounded Trailing}
In this section, we introduce our core notion of \emph{bounded 
trailing}, which we will prove to be an undecidable 
characterization of \fnft{}s having an equivalent \ftdfa in the 
following three sections.
We begin by recapitulating the notion of \emph{bounded variation}
introduced by Choffrut~\cite{CH79}, 
which characterizes \funnft{}s having an equivalent \dft.

For any two words $v_1, v_2$, we denote by $\textrm{lcp}(v_1,v_2)$
their longest common prefix and define the distance 
\[d(v_1, v_2)=|v_1|+|v_2|-2\cdot |\textrm{lcp}(v_1,v_2)|.\] 
Observe that 
prefixing any word $u$ to both $v_1$ and $v_2$ does not change 
their distance; we have $d(uv_1,uv_2)=d(v_1,v_2)$. 

\begin{definition}[Bounded variation]\label{def:bv}
A relation $R\subseteq\Sigma^\ast\times\Gamma^\ast$ has 
\emph{bounded variation} if 
\begin{flalign*}&\forall\, k\in\N\colon\, 
\exists\, \traillength\in\N\colon \quad\forall\, 
(a_1,u_1),(a_2,u_2)\in R\colon\,\\[.5ex]
  &d(a_1,a_2)\le k\quad\Longrightarrow\quad d(u_1,u_2)\le 
  \traillength.
\end{flalign*}
\end{definition}

If $a_1$ and $a_2$ have a common prefix $a$, we can split them into 
$a_1=ab_1$ and $a_2=ab_2$, respectively. 
We have $d(b_1,b_2)\le |b_1|+|b_2|$ and, if $a$ is the longest 
common prefix of $a_1$ and $a_2$, indeed $d(b_1,b_2)= |b_1|+|b_2|$. 
Applying $d(a_1,a_2)=d(ab_1,ab_2)=d(b_1,b_2)$, we thus obtain the 
following equivalent definition of bounded variation. 

\begin{lemma}[Bounded variation]\label{lem:bv}
A relation $R\subseteq\Sigma^\ast\times\Gamma^\ast$ has 
bounded variation if and only if 
\begin{flalign*}&\forall\, k\in\N\colon\, 
\exists\, \traillength\in\N\colon \quad\forall\, 
a, b_1,b_2\in\Sigma^\ast,\  
\,u_1,u_2\in\Gamma^\ast\colon\,\\[.5ex]
  &(ab_1,u_1)\in R\ \land\ (ab_2,u_2)\in R\ \land\ |b_1|+|b_2|\le 
  k\\
&\Longrightarrow\quad d(u_1,u_2)\le \traillength.
\end{flalign*}
\end{lemma}

We can now instantiate the characterization of bounded variation 
described in \cref{lem:bv} for the case of a relation recognized by 
a transducer. 
Specifically, we translate an input-output pair lying in 
the relation into the existence of a concrete computation by an 
\nft{}.
Appropriately splitting $u_1$ and $u_2$ into 
$uv_1w_1$ and $uv_2w_2$, respectively, with a 
(not necessarily longest)
common prefix $u$, and applying 
$d(u_1,u_2)=d(uv_1w_1,uv_2w_2)=d(v_1w_1,v_2w_2)$, this 
results in the following definition. 

\begin{definition}[Bounded variation for transducers]\label{def:bvt}
An \nft $T=(Q,\Sigma,\Gamma,q_0,F,\delta)$ has \emph{bounded 
variation} if 
\begin{flalign*}&\forall\, k\in\N\colon\, 
\exists\, \traillength\in\N\colon \quad\forall\,q_1,q_2\in Q,\ 
\,f_1,f_2\in F,&\\ 
&\hspace*{9.2em}\,a,b_1,b_2\in\Sigma^\ast,\  
\,u,
v_1,v_2,w_1,w_2\in\Gamma^\ast\colon\,\\[.5ex]
&q_0\tr{a}{uv_1}q_1\tr{b_1}{w_1}f_1\;\land\;
q_0\tr[0.9em]{a}{uv_2}q_2\tr[1.7em]{b_2}{w_2}f_2\;\land\;
|b_1|+|b_2|\leq k\\
&\Longrightarrow\quad d(v_1w_1, v_2w_2)\le \traillength.
\end{flalign*}
\end{definition}

We will now simplify \cref{def:bvt}. 
It follows from the definition of $d$ that 
\begin{align*}
d(v_1, v_2)-\big| |w_1|-|w_2| \big|&\le d(v_1w_1, v_2w_2)\\
&\le d(v_1, 
v_2)+|w_1|+|w_2|.
\end{align*}
(These bounds are tight; consider $v_1=000$, $v_2=0$, $w_1=0$ and 
either $w_2=000$ or $w_2=100$.) 
Using the bounded output speed \outputspeed of $T$ (\cref{def:os}),
we obtain 
\[||w_1|-|w_2||\le |w_1|+|w_2|\le 
\outputspeed(|b_1|+|b_2|)\le\outputspeed k.\] 
Hence, the difference between $d(v_1, v_2)$ and $d(v_1w_1, v_2w_2)$ 
can be bounded by a term depending only on $T$ and $k$ and can thus 
be hidden within the choice of $\traillength$. 
We can therefore simplify $d(v_1w_1, v_2w_2)$ to $d(v_1, v_2)$. 

Using the shortcut guarantee \shortcut of $T$ (\cref{def:sg}) for 
$b_1$ and $b_2$, we can now assume without loss of generality that 
$|b_1|+|b_2|\le 2\shortcut$ because replacing $b_1$ and $b_2$ can 
only change the respective output words $w_1$ and $w_2$, which are 
no longer referenced anywhere else. 
This finally allows us to drop the universally quantified variable 
$k\in\N$ altogether. 
We thus obtain the following equivalent definition of 
bounded variation.

\begin{lemma}\label{lem:bvtt}
An \nft $T=(Q,\Sigma,\Gamma,q_0,F,\delta)$ has 
bounded 
variation if and only if 
\begin{align*}
&\exists\, \traillength\in\N\colon\,
\quad\forall\,q_1,q_2\in Q,\ \,f_1,f_2\in F,\\ 
&\hspace*{5.5em}\,a,b_1,b_2\in\Sigma^\ast,\ \,u,v_1,v_2,w_1,w_2
\in\Gamma^\ast\colon\,\\[.5ex]
&q_0\tr{a}{uv_1}q_1\tr{b_1}{w_1}f_1\ \ \land\ \ 
q_0\tr[0.9em]{a}{uv_2}q_2\tr{b_2}{w_2}f_2\\
&\Longrightarrow\quad d(v_1, v_2)\le \traillength.
\end{align*}
\end{lemma}

Let us call a computation \emph{useful} if it can be extended to 
an accepting computation. 
Intuitively, an \nft{} has bounded variation if the outputs
of any two useful computations on the same input prefix $a$ 
only differ on suffixes of bounded length.

Our notion of bounded trailing relaxes the notion of 
bounded variation by only looking at certain pairs of 
useful computations on the same input prefix $a$, 
namely computation pairs such that 
the output $uv$ of the first computation contains 
the output $u$ of the second computation as a prefix  
and such that the second computation can be extended to 
an accepting one that \emph{trails} the first computation, 
in the sense of catching up with its output, by first
producing the missing suffix $v$.

Since we have only restricted the computation pairs for which 
we impose the condition on the output, we have indeed relaxed 
the notion of bounded variation. 
In particular, bounded variation implies bounded trailing.

We now formally define our notion of bounded trailing.
Starting from the condition in \cref{lem:bvtt}, we set 
$v_2=\lambda$, rename $v_1$ to $v$, and 
require that $w_2$ starts with the prefix $v$, which we then 
split off in the notation. 
Finally, we use $d(\lambda,v)=|v|$.

\begin{definition}[Bounded trailing]\label{def:bt}
An NFT $T=(Q,\Sigma,\Gamma,q_0,F,\delta)$ has 
\emph{bounded trailing} if 
\begin{align*}
&\exists\, \traillength\in\N\colon\,
\quad\forall\,q_1,q_2\in Q,\ \,f_1,f_2\in F,\\
&\hspace*{5.75em}a,b_1,b_2\in\Sigma^\ast,\ 
\,u,v,w_1,w_2\in\Gamma^\ast\colon\,\\[1.8ex]
&q_0\tr{a}{uv}q_1\tr{b_1}{w_1}f_1\ \ \land\ \ 
q_0\tr[0.9em]{a}{u}q_2\tr[1.7em]{b_2}{vw_2}f_2\ \ 
\Longrightarrow\ \ 
|v|\le
\traillength.
\end{align*}
\end{definition}
\Cref{ex:constr} describes a concrete transducer with bounded 
trailing.

The following three sections prove bounded trailing to be 
a necessary and sufficient condition for an \fnft to have 
an equivalent \tdfa and that bounded trailing is undecidable.
In contrast, bounded variation is decidable since it is equivalent 
to the so-called \emph{twinning property}, which is decidable in 
polynomial time~\cite{WK95}. 
We also point out that bounded trailing does not 
characterize having an equivalent \tdfa{} 
for \nft{}s that are not finite-valued. Constructing a 
counterexample is straightforward; we provide one in the following 
example.

\begin{example}\label{ex:counterexample}
We describe a not finite-valued transducer that has an equivalent 
\tdfa but nevertheless unbounded trailing. 

Consider the relation 
\[\{(0^n,0^m)\mid\,n\in\N\wedge m\in\N\wedge 0 
\le m\le 2n\},\]
which is clearly not finite-valued. 
It is easy for a \tdfa to recognize this relation by 
repeatedly advancing first the output head 
twice and then the input head once. Any input-output pair is in the 
language if and only if 
the end of the output is reached before the end of the input and 
all read symbols before the end markers are $0$. 
It is also simple to build an equivalent \nft{} with unbounded 
trailing as follows.
Whenever reading one input symbol, the transducer writes a $0$ to 
the output tape either zero, one, or two times. 

This transducer does indeed have unbounded trailing since for every 
even $\traillength\in\N$, it has the two 
computations 
\begin{align*}
\qquad q_0\tr{0^{t/2}}{\lambda 
0^t}q_\indexcaseA\tr{0^{t/2}}{\lambda}f_\indexcaseA\ \ 
\text{and}\ \ 
q_0\tr[0.9em]{0^{t/2}}{\lambda}q_\indexcaseB
\tr[1.7em]{0^{t/2}}{0^t\lambda
}f_\indexcaseB, 
\end{align*}
where $\lambda$ is the empty word and 
\[a=b_1=b_2=0^{\traillength/2}, u=w_1=w_2=\lambda,\text{ and 
}v=0^\traillength\] in \cref{def:bt}.
\end{example}

\section{Bounded Trailing is Sufficient}\label{sec:sufficiency}
In this section, we show that any \nft that 
has bounded trailing can be transformed
into an equivalent \tdfa.
Let $T$ be an \nft with a trailing bound $\traillength\in\N$. 
We construct an equivalent \tdfa $A$ that simulates 
all nondeterministic computations of transducer $T$ that 
are compatible with the output seen so far.
Automaton $A$ uses its states to maintain a subword $z$ of 
the output word with the following property.
For the currently read prefix $p$ of the input word,
there is a prefix $x$ of the output word such that for any 
accepting computation of $T$ that starts in 
the initial state $q_0$ of $T$, reads $p$, reaches a state $q$ 
of $T$, and produces an output $w$ consistent with the given 
output tape, we can write $w$ as $xy$ for a prefix $y$ of $z$. 
In other words: 
Whatever prefix $p$ of its input word automaton $A$ has read 
at any point, it has always stored a $z$ such that for some $x$ 
every accepting computation of $T$ on $p$ consistent 
with $A$'s output word has the form $q_0\tr {p}{xy}q$ 
for some prefix $y$ of $z$.

Automaton $A$ stores in its current state a representation of 
each such computation of $T$, namely the pair $(q, |y|)$. 
We denote the set of these pairs by $\pairset$. 
We show that storing $\pairset$ is feasible with a finite set 
of states by maintaining a subword $z$ of length 
at most $\buffersize=\outputspeed+\traillength$, 
where $\outputspeed$ and $\traillength$ are $T$'s output speed 
and trailing bound, respectively.
Initially, $z$ is empty and the set $\pairset$ of pairs $(q, n)$
stored in $A$'s state contains only a single pair, 
$(q_0, 0)$, where $q_0$ is $T$'s initial state. 
This reflects the fact that the only computation of $T$ on 
the empty prefix of the input tape keeps $T$ in 
its initial state and produces no output.
If $L(T)=\emptyset$, then $A$ immediately transitions 
into a rejecting sink state.

Automaton $A$ now proceeds as follows.
As long as the length of the subsequence $z$ stays 
below $\buffersize$ and the output tape has not been fully read,
the next symbol on the output tape is read and appended to $z$.
Moreover, $A$ removes from the set $\pairset$ all representations 
of computations that cannot be extended to an 
accepting computation with the extended subsequence $z$ 
of the output tape.
If the set $\pairset$ becomes empty, then automaton $A$ 
transitions into a rejecting sink state 
because there is no accepting computation of $T$
consistent with the given input and output tape.
Otherwise, $A$ determines the minimum $m$ such that 
$(q, m)\in \pairset$ for some $q$
and drops the first $m$ output symbols from the subsequence~$z$.
This corresponds to cutting off from $z$ a prefix of length $m$ 
and appending it to $x$.
Note that this is sound because none of the stored computations end 
before outputting the new $x$.
This way, $A$ maintains the invariant that there is 
some state $q$ of $T$ such that $(q, 0)\in \pairset$.

Once the length of the subsequence $z$ becomes $\buffersize$ or 
the output tape has been fully read,
automaton $A$ reads in the next symbol from the input tape and 
then simulates every single step that is nondeterministically 
possible for every single stored computation and updates 
the set $\pairset$ accordingly to a new $\pairset'$.
The mentioned invariant guarantees that $(q, 0)\in \pairset$ 
for some state $q$ of $T$.
The fact that $T$ has bounded trailing with a 
trailing bound $\traillength$ implies that $n\leq \traillength$
holds for every pair $(q, n)\in \pairset$.
Hence, performing one further nondeterministically possible step 
continuing $T$'s computation represented by $(q, n)\in\pairset$
yields a pair $(q', n')\in\pairset'$ that satisfies
$n'\le\traillength+\outputspeed$ since $\outputspeed$ is the 
longest output that $T$ can produce while reading a single symbol.
This is just within the length limit 
$\buffersize=\outputspeed+\traillength$ that we set for $z$.
As before, automaton $A$ rejects 
if the set $\pairset'$ becomes empty, 
because this means there is no accepting computation 
of $T$ consistent with the given input and output tape.
Otherwise, automaton $A$ performs on $\pairset'$ 
the normalization described in the previous paragraph
to obtain a new set $\pairset$ that maintains the invariant 
that there is some $q'$ for which $(q', 0)\in \pairset$.

Finally, if both the input and output tape have been fully read, 
automaton $A$ accepts if and only if
there is some $q\in F$ with $(q, |z|)\in \pairset$, that is, 
an accepting state $q$ of $T$ that some computation of $T$ arrives 
at after producing an output that matches $xz$ until the very end, 
meaning that the output is equal to the content of the output tape. 

\begin{example}\label{ex:constr}
We now provide a concrete example of the general construction 
described above; that is, we use a given \nft with bounded trailing 
to construct an equivalent \tdfa. 

We consider the \nft{} $T$ with the set of states $Q=\{q_0, q_1, 
q_2, q_3\}$,
initial state $q_0\in Q$, the set of accepting states 
$F=\{q_0\}$,
and the transition function $\delta$ given in \cref{fig:delta}.
Let us define the language
\[L_0=\{(aa, ababa), (aa, ababab), (ab, ababaa)\}.\]
One can check that the language computed by $T$ is the Kleene 
closure of $L_0$ (i.e., $L(T)=L_0^*$) and that the trailing of $T$ 
is bounded by $\traillength=1$.

The computation of automaton $A$ on the input-output pair $(aa, 
ababab)$ is summarized in \cref{fig:construction}. We now describe 
this computation in detail. 
Because $L(T)\neq\emptyset$, the initial state of automaton $A$ is 
storing $z=\lambda$, where $\lambda$ denotes the empty word, and 
$\pairset=\{(q_0, 0)\}$. 
The output speed of $T$---that is, the length of a longest output 
that $T$ can produce while reading one input symbol---%
is $\outputspeed=4$.
Hence, the maximum length of the subsequence $z$ maintained in 
$A$'s state is $\buffersize=\outputspeed+\traillength=5$.

Because the output tape consists of six symbols, the first 
$\buffersize=5$ of them are read and appended to  $z$.
The nonempty output words in $L(T)$ all start with the same five 
symbols,
and the only pair in $\pairset$, namely $(q_0, 0)$, can be extended 
to an accepting computation consistent with next five output 
symbols.
These symbols are therefore successively output and appended to 
$z$, while $(q_0, 0)$ is being kept in $\pairset$. 
After this, $A$'s state is storing 
$z=ababa$ and $\pairset=\{(q_0, 0)\}$.

Now that the length of $z$ is $5=\buffersize$, the first symbol $a$ 
from the input tape is read.
All nondeterministic steps from
\[\delta(q_0, a)=\{(q_1, aba), (q_2, 
abab), (q_3, abab)\}\]
are consistent with $z$, which leads to 
$\pairset'=\{(q_1, 3), 
(q_2, 4), (q_3, 4)\}$. 
The minimum $m$ for which $(q, m)\in \pairset'$ for some $q$ is 
$m=3$.
After performing the normalization of $\pairset'$ with this $m$, 
the state of automaton $A$ is storing $z=ba$ and 
$\pairset=\{(q_1, 
0), (q_2, 1), (q_3, 1)\}$.

Since the length of $z$ is now $2<\buffersize$, automaton $A$ reads 
the next output symbol $b$
and appends it to $z$, which thus becomes $z=bab$.
Then $A$ removes the pair $(q_3, 1)$ from the set $\pairset$ 
because this pair can no longer be extended to an accepting 
computation consistent
with the extended subsequence $z=bab$ since the only possible 
transition from $q_3$ produces the output $aa$, which is 
inconsistent
with the suffix $ab$ of $z=bab$. 
The remaining two pairs $(q_1, 0)$ and $(q_2, 1)$ are still 
consistent with the extended $z=bab$.
As $(q_1, 0)$ stays in $\pairset$, no normalization need be 
performed.%

Now that the end of the output tape has been reached, $A$ reads in 
the next symbol $a$
from the input tape despite $|z|=3<\buffersize$. 
Performing one step that reads $a$ on every pair in $\pairset$ 
yields $\pairset'=\{(q_0, 3)\}$.
Note that $(q_0, 2)\not\in \pairset'$ because each transition from 
$q_0$ starts with a $b$, the last symbol of $z$.
After one more normalization, $A$'s state is storing $z=\lambda$ 
and $\pairset=\{(q_0, 0)\}$.

Finally, $A$ has reached the end of both the input and output tape, 
thus it checks whether $(q, |z|)\in \pairset$ for some $q\in F$.
Because $(q_0, 0)\in \pairset$, $q_0\in F$, and $|z|=0$, automaton 
$A$ accepts the 
input-output pair $(aa, ababab)\in L(T)$.

\begin{figure}
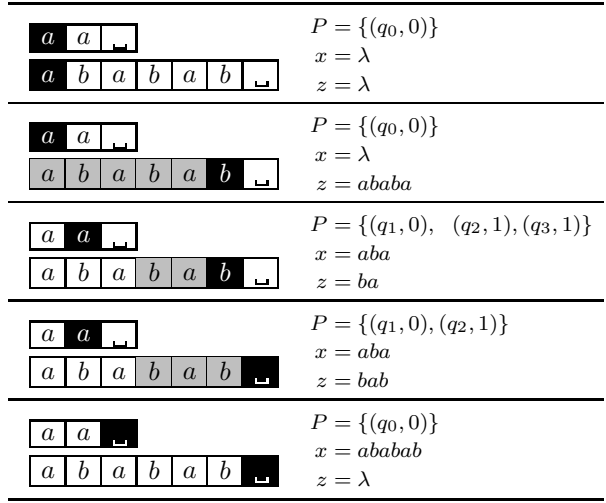

\begin{subfigure}{\linewidth}
\centering
\begin{tabular}{ccc}
\toprule
&$a$&$b$\\
\cmidrule(lr){2-2}\cmidrule(lr){3-3}
$q_0$&$\{(q_1, aba), (q_2, abab), (q_3, abab)\}$&$\emptyset$\\
$q_1$&$\{(q_0, ba), (q_0, bab)\}$&$\emptyset$\\
$q_2$&$\{(q_0, ab)\}$&$\emptyset$\\
$q_3$&$\emptyset$&$\{(q_0,aa)\}$\\
\bottomrule
\end{tabular}
\caption{The transition function $\delta$ of the given \nft{} 
$T$.\hfill}\label{fig:delta}
\end{subfigure}

\renewcommand\arraystretch{.7} \setlength\minrowclearance{1.4pt}
\setlength\tabcolsep{4pt} 

\newcommand{\head}[1]{\cellcolor{black}\textcolor{white}{$#1$}}
\newcommand{\buf}[1]{\cellcolor{black!25}$#1$}
\newcommand{\myvrule}{{\vrule depth 16pt height 20pt width 0pt}}

\newcommand{\cellwidth}{.5em}
\vspace*{2em}
\begin{subfigure}{\linewidth}
\centering
\begin{tabular}{ll}
\toprule
\begin{tabular}{l}
\begin{tabular}{|*3{C{\cellwidth}|}}
\hline
\head{a}&$a$&$\blank$\\
\hline
\end{tabular}\\[.4ex]
\begin{tabular}{|*7{C{\cellwidth}|}}
\hline
\head{a}&$b$&$a$&$b$&$a$&$b$&$\blank$\\
\hline
\end{tabular}
\end{tabular}
&
\footnotesize $\begin{aligned}
P&=\{(q_0, 0)\}\\[-.3ex]
x&=\lambda\\[-.3ex]
z&=\lambda
\end{aligned}$
\\
\midrule
\begin{tabular}{l}
\begin{tabular}{|*3{C{\cellwidth}|}}
\hline
\head{a}&$a$&$\blank$\\
\hline
\end{tabular}\\[.4ex]
\begin{tabular}{|*7{C{\cellwidth}|}}
\hline
\buf a&\buf b&\buf a&\buf b&\buf a&\head b&$\blank$\\
\hline
\end{tabular}
\end{tabular}
&
\footnotesize $\begin{aligned}
P&=\{(q_0, 0)\}\\[-.3ex]
x&=\lambda\\[-.3ex]
z&=ababa
\end{aligned}$
\\
\midrule
\begin{tabular}{l}
\begin{tabular}{|*3{C{\cellwidth}|}}
\hline
$a$&\head a&$\blank$\\
\hline
\end{tabular}\\[.4ex]
\begin{tabular}{|*7{C{\cellwidth}|}}
\hline
$a$&$b$&$a$&\buf b&\buf a&\head b&$\blank$\\
\hline
\end{tabular}
\end{tabular}
&
\footnotesize $\begin{aligned}
P&=\{(q_1, 0), \ \;(q_2, 1),(q_3, 1)\}\\[-.3ex]
x&=aba\\[-.3ex]
z&=ba
\end{aligned}$
\\
\midrule
\begin{tabular}{l}
\begin{tabular}{|*3{C{\cellwidth}|}}
\hline
$a$&\head a&$\blank$\\
\hline
\end{tabular}\\[.4ex]
\begin{tabular}{|*7{C{\cellwidth}|}}
\hline
$a$&$b$&$a$&\buf b&\buf a&\buf b&\head\blank\\
\hline
\end{tabular}
\end{tabular}
&
\footnotesize $\begin{aligned}
P&=\{(q_1, 0),(q_2, 1)\}\\[-.3ex]
x&=aba\\[-.3ex]
z&=bab
\end{aligned}$
\\
\midrule
\begin{tabular}{l}
\begin{tabular}{|*3{C{\cellwidth}|}}
\hline
$a$&$a$&\head\blank\\
\hline
\end{tabular}\\[.4ex]
\begin{tabular}{|*7{C{\cellwidth}|}}
\hline
$a$&$b$&$a$&$b$&$a$&$b$&\head\blank\\
\hline
\end{tabular}
\end{tabular}
&
\footnotesize $\begin{aligned}
P&=\{(q_0, 0)\}\\[-.3ex]
x&=ababab\\[-.3ex]
z&=\lambda
\end{aligned}$\\
\bottomrule
\end{tabular}
\caption{The computation of \tdfa{} $A$ on input $(aa, 
ababab)$. The positions of the heads are marked in black, the 
subsequence $z$ of the output tape is shaded gray, and $x$ consists 
of the white cells before $z$.}
\label{fig:construction}
\end{subfigure}
\caption{Illustrations pertaining to \cref{ex:constr}.}
\end{figure}
\end{example}

We conclude this section by formally stating its main result. 
\begin{theorem}\label{thm:bt-sufficient}
Any \nft{} with bounded trailing has an equivalent \tdfa{}.
\end{theorem}

\section{Bounded Trailing is Necessary}\label{sec:bt-necessary}

In this section, we prove the reverse of \cref{thm:bt-sufficient} 
for the case of finite-valuedness. 

\begin{theorem}\label{thm:bt-necessary}
Any \fnft{} with an equivalent \tdfa{} has bounded trailing.
\end{theorem}

\begin{proof}
Let an \fnft $T$ and a \tdfa $A$ with $L(T)=L(A)$ be given. 
(We may assume that $T$ and $A$ use a common 
input alphabet $\Sigma$ and a common output alphabet $\Gamma$.) 
Let $k$ be an arbitrary integer such that $T$ is $k$-valued. 
We assume that $T$ has unbounded trailing and derive from this 
a contradiction to the $k$-valuedness of $T$, 
thus proving the theorem. 
We may assume without loss of generality that $A$ moves 
exactly one head in each step because any step moving 
both heads at once can be simulated by two steps moving 
only one head at a time using one additional intermediate state. 

Denote the state sets of transducer $T$ and automaton $A$ 
by $Q_T$ and $Q_A$, respectively. 
Let $\shortcut$ and $\outputspeed$ be the transducer's 
shortcut guarantee and output speed, respectively. 
Finally, we define a \emph{homestretch length} 
$\homestretch=(\shortcut+1)\cdot(|Q_A|+1)$ and 
a \emph{trail length} minimum $\traillength = 
\homestretch+\traillength_1+\dots+\traillength_k$, where 
$\traillength_i$ is recursively defined by 
\[\traillength_i=2\outputspeed+\outputspeed(1+
i(\traillength_0+\traillength_1+\dots+\traillength_{i-1}))\cdot
(1+|Q_T|\cdot|Q_A|^i)\]
with the base case
$\traillength_0=\homestretch+(\shortcut+1)\outputspeed$. 
The reason for choosing exactly these values for \homestretch 
and \traillength will become clear during the proof. 
For now, we note that they depend only on the given transducer $T$ 
and the automaton $A$, hence they are well-defined and 
fixed within this proof.  

Since $T$ has unbounded trailing, it has two accepting computations 
\begin{align*}
q_0\tr{a}{uv}{}
q_\indexcaseA\tr{b_\indexcaseA}{w_\indexcaseA}f_\indexcaseA
\quad\text{ and }\quad
q_0\tr[1.2em]{a}{u}{}
q_\indexcaseB\tr{b_\indexcaseB}{vw_\indexcaseB}f_\indexcaseB
\end{align*}
with $|v|>\traillength$, where $q_0$ is the initial state and 
$f_\indexcaseA$ and $f_\indexcaseB$ are accepting states of $T$. 
We decompose the \emph{trail} $v$ of length at least 
$\traillength$ as $\pre vv_kv_{k-1}\dots v_2v_1\suf v$ with 
$|v_k|=\traillength_k,\dots,|v_1|=\traillength_1$ and 
$|\suf v|=\homestretch$. We call $\suf v$ the \emph{homestretch}. 
We consider the input-output prefix $(a,uv)$ common to both 
computations of automaton $A$. 
The two-tape automaton model ensures that the two heads of $A$ 
will eventually reach the end of $a$ and $v_1$, respectively. 
Depending on which one does so first, we distinguish two cases.

\def\tapeheight{0.45}
\def\separator{0.48}

\newcommand{\bottomtape}[3]{\draw[semithick] (#1,0) rectangle
(#2,\tapeheight) node[midway]
{#3\vphantom{$\strut_{\strut}^{\strut}$}};}
\newcommand{\toptape}[3]{\draw[semithick]
(#1,\tapeheight+\separator) rectangle (#2,2*\tapeheight+\separator)
node[midway] {#3\vphantom{$\strut_{\strut}^{\strut}$}};}

\newcommand{\bottomhead}[2][white]{\draw[thick,fill=#1]
(#2,\tapeheight+.02)--+(.15,.3)--+(-.15,.3)--cycle;}
\newcommand{\tophead}[2][white]{\draw[thick,fill=#1]
(#2,\tapeheight+\separator-.02)--+(.15,-.3)--+(-.15,-.3)--cycle;}

\newcommand{\halftophead}[2][white]{\draw[thick,fill=#1]
(#2,\tapeheight+\separator-.02)--+(0,-.3)--+(-.15,-.3)--cycle;}

\newcommand{\halftopheadright}[2][white]{\draw[thick,fill=#1]
(#2,\tapeheight+\separator-.02)--+(0,-.3)--+(.15,-.3)--cycle;}

\newcommand{\halftopheadalt}[2][white]{\draw[ultra thin,fill=#1]
(#2,2*\tapeheight+\separator+.02)--+(-.15,.3)--+(0,.3)--cycle;}

\newcommand{\halftopheadaltright}[2][white]{\draw[ultra 
thin,fill=#1]
(#2,2*\tapeheight+\separator+.02)--+(.15,.3)--+(0,.3)--cycle;}

\newcommand{\patterntophead}[2][north west lines]{\draw[thick,
pattern=#1]
(#2,\tapeheight+\separator-.02)--+(.15,-.3)--+(-.15,-.3)--cycle;}

\newcommand{\bottomheadalt}[2][white]{\draw[thick,fill=#1]
(#2,-.02)--+(.15,-.3)--+(-.15,-.3)--cycle;}
\newcommand{\topheadalt}[2][white]{\draw[thick,fill=#1]
(#2,2*\tapeheight+\separator+.02)--+(.15,.3)--+(-.15,.3)--cycle;}
\newcommand{\topheadhigh}[2][white]{\draw[thick,fill=#1]
(#2,2*\tapeheight+\separator+.02+.3+.02)--+(.15,.3)--+(-.15,.3)--cycle;}

\newcommand{\bottombrace}[4][0]{
\draw [semithick, decorate, 
decoration={calligraphic brace, amplitude=6pt, mirror}]
(#2+.02,0-.1+#1) -- (#3-.02,0-.1+#1) node[draw=none, fill=none, 
midway, below, yshift=-1ex] {#4};
}

\newcommand{\topbrace}[4][0]{
\draw [semithick, decorate, decoration={calligraphic brace, 
amplitude=6pt, mirror}]
(#3-.02,2*\tapeheight+\separator+.08+#1) -- 
(#2+.02,2*\tapeheight+\separator+.08+#1) node[draw=none, fill=none,
midway, above, yshift=1ex] {#4};
}

\begin{figure*}[ht]
\begin{subfigure}{\textwidth}
\begin{tikzpicture}[scale=1.25]
\toptape{0}{9}{$a$}
\toptape{9}{10.7}{$b_\indexcaseA$}

\bottomtape{0}{2}{$u$}
\bottomtape{2}{11}{$v$}
\bottomtape{11}{11.8}{$w_\indexcaseA$}

\tophead{9}
\halftopheadright[black]{9}
\bottomhead{1.3}
\bottomhead[black]{11}
\end{tikzpicture}
\caption{Case \caseA: Automaton $A$'s head positions when 
its input head has just reached the end of $a$ and 
transducer $T$'s head positions after its computation 
$q_0\tr a{uv}q_\indexcaseA$.}
\label{fig:case1}
\end{subfigure}
\begin{subfigure}{\textwidth}
\begin{tikzpicture}[scale=1.25]
\toptape{0}{9}{$a$}
\toptape{9}{13.5}{$b_\indexcaseB$}

\bottomtape{0}{2}{$u$}
\bottomtape{2}{2.6}{$\pre v$}
\bottomtape{2.6}{5.4}{$v_k$}
\node at (6.1,0.2) {$\cdots$};
\bottomtape{6.8}{8.6}{$v_2$}
\bottomtape{8.6}{10}{$v_1$}
\bottomtape{10}{11}{$\suf v$}
\bottomtape{11}{12.7}{$w_\indexcaseB$}

\tophead{5.5}
\tophead[black]{9}
\bottomhead[black]{2}
\bottomhead{10}

\bottombrace{2.6}{5.4}{$\traillength_k$}
\bottombrace{6.8}{8.6}{$\traillength_2$}
\bottombrace{8.6}{10}{$\traillength_1$}
\bottombrace{10}{11}{$\homestretch$}
\end{tikzpicture}
\caption{Case \caseB: Automaton $A$'s head positions when its 
output head has just reached the end of $v_1$ and transducer $T$'s
head positions after computing $q_0\tr{a}{u}q_\indexcaseB$.}
\label{fig:case2}
\end{subfigure}
\caption{Two different input-output pairs in the relation 
$L(A)=L(T)$. 
The black and white triangles represent the heads of 
transducer $T$ and automaton $A$, respectively. 
The second subfigure shows a decomposition 
$\pre vv_kv_{k-1}\dots v_2v_1\suf v$ of $v$.}
\end{figure*}
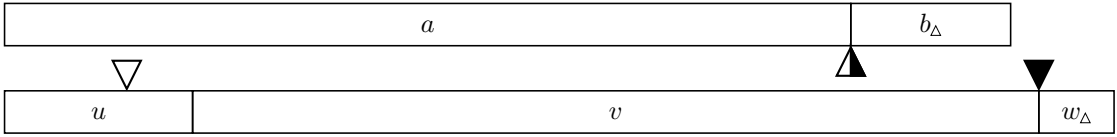
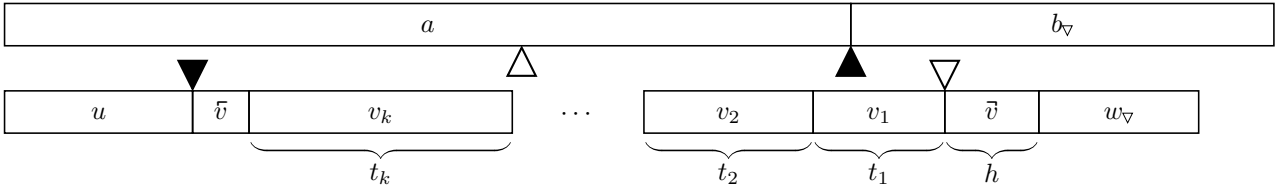
\medskip
\textbf{Case \boldcaseA: Input head is first.}
In this case, we consider the input-output pair 
$(ab_\indexcaseA,uvw_\indexcaseA)$; see \cref{fig:case1}. 

We begin by showing why we may assume without loss of generality 
that $|b_\indexcaseA|\le \shortcut$. 
For this, we consider transducer $T$'s configuration after 
the computation $q_0\tr{a}{uv}{} q_\indexcaseA$; 
the corresponding head positions are indicated 
by the black triangles in \cref{fig:case1}. 
The shortcut guarantee $\shortcut$ ensures the existence of 
a word pair $(b,w)\in\Sigma^\ast\times\Gamma^\ast$ and 
an accepting state $f\in F$ such that 
$q_\indexcaseA\tr[1.0em] {\raisebox{-.2ex}{$b$}}wf$ and 
$|b|\le\shortcut$; we can therefore substitute $b$, $w$, and $f$
for $b_\indexcaseA$, $w_\indexcaseA$, and $f_\indexcaseA$ 
if necessary and thus assume $b_\indexcaseA$. 

Now, we consider automaton $A$'s computation on the same 
input-output pair $(ab_\indexcaseA,uvw_\indexcaseA)$. 
Let $x$ denote the output word's suffix that has not yet been read 
by $A$ when its input head has just reached the end of $a$; 
see the white triangles in \cref{fig:case1} for 
$A$'s head positions. 
The input head has only $b_\indexcaseA$ left to read 
but the output head all of $x$.  
We have already established that 
$|b_\indexcaseA|\le \shortcut$ using the shortcut guarantee, 
and we know that 
$|x|\ge |\suf v|=\homestretch$ since $x$ contains 
the homestretch $\suf v$. 
In each step, $A$ advances exactly one head by exactly one symbol, 
and a head does not move anymore once it has reached the end 
of the word written on its tape. 
Thus at most $\shortcut$ movements remain for the input head 
but at least $\homestretch$ for the output head. 
We call a step in which the input head moves an \emph{input step} 
and a step in which the output head moves an \emph{output step}. 
The input steps split the remaining computation into 
at most $\shortcut+1$ sequences of consecutive output steps. 
Since there are at least $\homestretch$ output steps, 
there is at least one sequence of $\homestretch/(\shortcut+1)$ 
uninterrupted output steps.
Because of $\homestretch/(\shortcut+1)>|Q_A|$, 
this sequence contains at least two different output steps 
leading $A$ into the same state. 
Choosing any two such steps, we can repeat a nonempty part of 
the output word arbitrarily often, namely, the part 
that starts at the position of the output head 
immediately after the first step and 
ends with the symbol at the position of the output head 
just before the second step. 
This results in arbitrarily many accepting computations, 
with the word on the input tape unmodified. 
Hence $A$ associates infinitely many different output words 
with the same input word $ab_\indexcaseA$, 
contradicting the $k$-valuedness of $L(A)=L(T)$. 

\medskip
\textbf{Case \boldcaseB: Output head is first.} In this case, 
automaton~$A$'s output head reaches the end of $v_1$ 
before its input head has 
finished reading $a$. This remains true for $A$'s computation on 
the input-output pair $(ab_\indexcaseB,uvw_\indexcaseB)$; see 
\cref{fig:case2}. 

In what follows, we establish an upper bound on the length of 
the remaining output $\suf vw_\indexcaseB$. 
The homestretch length $|\suf v|=\homestretch$ is already fixed. 
The length of $w_\indexcaseB$ can be bounded by combining 
the shortcut guarantee $\shortcut$ and 
the output speed $\outputspeed$ as follows. 
Consider the transducer's computation 
\[q_\indexcaseB\tr{b_\indexcaseB}{vw_\indexcaseB}f_\indexcaseB.\]
 
Since the output head can write at most $\outputspeed$ symbols per 
step, we can cut off from $b_\indexcaseB$ a prefix $\pre 
b_\indexcaseB$ such that 
\[q_\indexcaseB\tr{\raisebox{-.2ex}{$\pre 
b_\indexcaseB$}}{\raisebox{-.4ex}{$v\pre 
w_\indexcaseB$}}q'\]
for some prefix $\pre w_\indexcaseB$ of $w_\indexcaseB$ with length 
$|\pre w_\indexcaseB|<\outputspeed$ and some state $q'$. 
Denote by $\suf b_\indexcaseB$ and $\suf w_\indexcaseB$ 
the remaining suffixes such that 
$b_\indexcaseB=\pre b_\indexcaseB\suf b_\indexcaseB$ and 
$w_\indexcaseB=\pre w_\indexcaseB\suf w_\indexcaseB$. 
The state $q'$ is co-reachable as evidenced by the 
accepting computation 
\[q_0\tr[0.9em] auq_\indexcaseB\tr{\raisebox{-.2ex}{$\pre 
b_\indexcaseB$}}{v\pre 
w_\indexcaseB\vphantom{w'}}q'\tr{\raisebox{-.2ex}{$\suf 
b_\indexcaseB$}}{\suf w_\indexcaseB\vphantom{w'}}f_\indexcaseB.\]
Using the definition of the shortcut guarantee $g$, 
we can therefore assume $|\suf b_\indexcaseB|\le\shortcut$ 
by substituting a sufficiently short $b'$ and some 
appropriate $w'$ and $f'$ for $\suf b_\indexcaseB$, 
$\suf w_\indexcaseB$, and $f_\indexcaseB$ if necessary. 
From this, we then obtain 
$|\suf w_\indexcaseB|\le \shortcut\cdot\outputspeed$ 
since the output speed \outputspeed tells us how many symbols 
the transducer can output at most when reading one input symbol. 
We can thus assume that 
\[|w_\indexcaseB|=|\pre w_\indexcaseB\suf 
w_\indexcaseB|<\outputspeed+\shortcut\cdot\outputspeed=(\shortcut+1)\outputspeed\]
 without loss of generality. This finally yields the desired upper 
bound $|\suf vw_\indexcaseB|< 
\homestretch+(\shortcut+1)\outputspeed$. 
 
Recall that $\pre vv_kv_{k-1}\dots v_2v_1\suf v$ is 
the unique decomposition of $v$ with
$|v_k|=\traillength_k,\dots,|v_1|=\traillength_1$ and 
$\suf v=\homestretch$.
Since $\traillength_0 = \homestretch + (\shortcut+1)\outputspeed$, 
we immediately obtain \[|v_iv_{i-1}\dots v_1\suf vw_\indexcaseB|\le 
\traillength_0+\traillength_1+\dots+\traillength_i.\] 
Let $\suf a$ be $a$'s suffix that is still unread at the moment 
when the output head of the automaton reaches the start of $\suf v$.
Denote automaton $A$'s initial state by $\hat q_0$ and choose any 
decomposition of $a$ into $\pre aa_ka_{k-1}\dots a_2a_1\suf a$ 
such that $A$'s computation $\hat 
q_0\tr[4em]{ab_\indexcaseB}{uvw_\indexcaseB}f_\indexcaseB$ splits 
into 
\[\hat q_0\tr{\pre a}{u\pre v}\hat q_{k+1}\tr{a_k}{v_k}\hat 
q_k\tr[2em]{a_{k-1}}{v_{k-1}}\dots\tr{a_2}{v_2}\hat 
q_2\tr{a_1}{v_1}\hat q_1\tr{\suf ab_\indexcaseB}{\suf 
vw_\indexcaseB}f_\indexcaseB.\]

We will now prove the following claim by induction over 
$i\in\{0,1,\dots,k\}$:\hfill

\begin{claim}\label{clm:induction}
For every $i\in\{0,1,\dots,k\}$, there is an input word 
$b_{i+1}\in\Sigma^\ast$ and, for every $j\in\{1,2,\dots,i\}$, 
an output word $\tilde v_j\in\Gamma^\ast$ with $|\tilde 
v_j|<|v_j|$, such that transducer $T$ has the accepting computation
\[q_0\tr[1.2em]{a}{u}{}q_\indexcaseB\tr[16em]{b_{i+1}}{\pre 
vv_kv_{k-1}\dots v_{i+1}\tilde v_i\tilde v_{i-1}\dots\tilde 
v_2\tilde v_1\suf vw_\indexcaseB}f_\indexcaseB\]
and such that automaton $A$ has $i+1$ computations 
that all read the same input word $ab_{i+1}$ overall, 
all begin with 
\[\hat q_0\tr[10em]{\pre aa_ka_{k-1}\dots a_{i+1}}{u\pre 
vv_kv_{k-1}\dots v_{i+1}}{}\hat q_{i+1},\] 
and from $\hat q_{i+1}$ lead to accepting states while reading 
the remaining input $a_i\dots a_2a_1\suf ab_i$ and 
the $i+1$ distinct outputs 
\begin{align*}
&v_iv_{i-1}\dots v_3v_2v_1,\\
&v_iv_{i-1}\dots v_3v_2\tilde v_1,\\
&v_iv_{i-1}\dots v_3\tilde v_2\tilde v_1,\\
&\dots,\\
&v_i\tilde v_{i-1}\dots\tilde v_3\tilde v_2\tilde v_1, \text{ and}\\
&\tilde v_i\tilde v_{i-1}\dots\tilde v_3\tilde v_2\tilde v_1,
\end{align*} respectively.
\end{claim}
Instantiated for $i=k$, this claim shows that 
the relation $\mathcal{L}(A)=\mathcal{L}(T)$ 
associates a single input word $ab_{k+1}$ with $k+1$ distinct 
output words of pairwise different lengths, 
contradicting the $k$-valuedness and thus concluding 
the proof of \cref{thm:bt-necessary}.

We now briefly outline the proof of \cref{clm:induction}; the full 
proof is provided afterward. 
Proving the induction basis is trivial: For $i=0$, 
the claim coincides precisely with the situation shown 
in \cref{fig:case2} when setting $b_1=b_\indexcaseB$. 
It remains to prove the induction step.
Assuming the claim for $i-1$ as our induction hypothesis, 
our goal is to find nonempty subwords $\loo b_i$ of $b_i$ and 
$\loo v_i$ of $v_i$ such that not only transducer $T$'s 
accepting computation loops on them, 
but all of $A$'s pairwise distinct accepting computations 
on the $i$ different input-output pairs loop on the common 
input subword $\loo b_i$ without advancing the output head at all. 
This is achievable since the length $t_i$ of the subword $v_i$ in 
the output is recursively defined to be sufficiently 
large---allowing us to force transducer $T$ into producing plenty 
of output while reading the input $b_i$---and simultaneously small 
enough---ensuring that automaton $A$ has very few steps that 
advance the output head while reading $b_i$. 
\end{proof}
We now prove \cref{clm:induction} used in the proof above.
\begin{proof}[Proof of \cref{clm:induction}]
Establishing the induction basis is trivial: For $i=0$, the claim 
coincides precisely with the situation shown in \cref{fig:case2} 
when setting $b_1=b_\indexcaseB$. 
It remains to prove the induction step. 
We fix an arbitrary $i\in\{1,\dots,k\}$, assume the claim for 
$i-1$ 
as our induction hypothesis, and prove it for $i$. 
First consider $T$'s computation given by the induction hypothesis 
for $i-1$. 
Because $T$ produces at most \outputspeed output symbols when 
reading one input symbol, there are decompositions $b_i=\pre 
b_i\inf b_i\suf b_i$ and $v_i=\pre v_i\inf v_i\suf v_i$ with $|\pre 
v_i|<\outputspeed$ and $|\suf v_i|<\outputspeed$ such that $T$ has 
a computation 
 \begin{align*}&q_0\tr[1.2em]{a}{u}{}q_\indexcaseB\tr[8em]{\pre 
 b_i}{\pre vv_kv_{k-1}\dots v_{i+1}\pre v_i}\pre q_i
\tr{\inf b_i}{\inf v_i}\suf q_i\\
&\hspace*{14.8em}\tr[8em]{\suf b_i}{\suf v_i\tilde 
v_{i-1}\dots\tilde 
v_1\suf vw_\indexcaseB}f_\indexcaseB.\end{align*}
Note that \[|\inf 
v_i|>\traillength_i-2\outputspeed=\outputspeed\cdot(1+
i(\traillength_0+\traillength_1+\dots+\traillength_{i-1}))\cdot
(1+|Q_T|\cdot|Q_A|^i).\] 
Thus $\pre q_i\tr[1em]{\inf{b}_i}{\inf{v}_i}\suf q_i$ contains more 
than 
\[(1+i(\traillength_0+\traillength_1+\dots+\traillength_{i-1}))\cdot
(1+|Q_T|\cdot|Q_A|^i)\] steps during which $T$ produces nonempty 
output. 
We call the positions of $\inf b_i$ at which $T$ produces nonempty 
output during its computation \emph{$T$-productive}. 

Now consider the $i$ different computations of $A$ given by the 
induction hypothesis for $i-1$. 
They all start with a common prefix that splits into 
\[\hat 
q_0\tr[8em]{\pre aa_ka_{k-1}\dots a_{i+1}}{u\pre vv_kv_{k-1}\dots 
v_{i+1}}{}\hat q_{i+1}\tr{a_i}{v_i}{}\hat q_i.\] 
Note that in the remainder of these computations, automaton $A$ 
always scans the same input word $a_{i-1}\dots a_2a_1b_i$ but $i$ 
different output words. These output words are indeed pairwise 
different since their lengths are 
\begin{align*}
|\tilde v_i\tilde v_{i-1}\tilde v_{i-2}\dots\tilde v_3\tilde 
v_2\tilde v_1\suf 
vw_\indexcaseB|{}&<|v_i\tilde v_{i-1}\tilde v_{i-2}\dots \tilde 
v_3\tilde v_2\tilde v_1\suf 
vw_\indexcaseB|\\
&<|v_iv_{i-1}\tilde v_{i-2}\dots \tilde v_3\tilde v_2\tilde v_1\suf 
vw_\indexcaseB|\\
&<\dots\\
&<|v_iv_{i-1}v_{i-2}\dots v_3\tilde v_2\tilde v_1\suf 
vw_\indexcaseB|\\
&<|v_iv_{i-1}v_{i-2}\dots v_3v_2\tilde v_1\suf 
vw_\indexcaseB|\\
&<|v_iv_{i-1}v_{i-2}\dots v_3v_2v_1\suf vw_\indexcaseB|\\
&\le 
\traillength_0+\traillength_1+\dots+\traillength_i.
\end{align*}

Since these are $i$ different computations, each of which has at 
most $\traillength_0+\traillength_1+\dots+\traillength_{i-1}$ 
output symbols, there are within the remaining $a_{i-1}\dots 
a_2a_1b_i$ input suffix at most 
$i(\traillength_0+\traillength_1+\dots+\traillength_{i-1})$ 
positions such that, for at least one of the $i$ computations by 
$A$, the input head of $A$ stays put in said position while its 
output head advances. We call these positions \emph{$A$-productive} 
and the others \emph{$A$-unproductive}. 
It follows in particular that the subword $\inf b_i$ in $b_i=\pre 
b_i\inf b_i\suf b_i$ contains at most 
\[i(\traillength_0+\traillength_1+\dots+\traillength_{i-1})\] 
positions that are $A$-productive. 
They delimit at most 
\[1+i(\traillength_0+\traillength_1+\dots+\traillength_{i-1})\] 
sequences of consecutive $A$-unproductive positions in $\inf b_i$. 
Using the lower bound on the number of $T$-productive positions in 
$\inf b_i$ derived above, we conclude that at least one of these 
$A$-unproductive sequences contains more than 
\[\frac{(1+i(\traillength_0+\traillength_1+\dots+\traillength_{i-1}))\cdot
|Q_T|\cdot|Q_A|^i}{1+i(\traillength_0+\traillength_1+\dots+\traillength_{i-1})}=
 |Q_T|\cdot|Q_A|^i\] 
$T$-productive positions. 
Thus $\inf b_i$ has a decomposition $\pre b_i'\loo b_i\suf b_i'$ 
such that, on the one hand, all $i$ computations of $A$ given by 
the induction hypothesis for $i-1$ loop through $\loo b_i$ 
simultaneously without advancing the output head at all and, on the 
other hand, $T$'s partial computation %
on $\inf b_i$ decomposes into 
\[\pre q_i\tr{\pre b_i'}{\pre v_i'}\loo q_i\tr{\loo b_i}{\loo 
v_i}\loo q_i\tr{\suf b_i'}{\suf v_i'}\suf q_i\] 
with a loop that starts at some state $\loo q_i$ and, while reading 
the input sequence $\loo b_i$, produces some output $\loo v_i$ that 
is nonempty due to at least one $T$-productive position in $\loo 
b_i$. 
Removing the nonproductive loop $\loo b_i$ from all $i$ of $A$'s 
computations given by the induction hypothesis for $i-1$, we obtain 
all but the last of $A$'s computations in the claim for $i$. 
Moreover, deleting the loop on $\loo b_i$ %
from $T$'s computation in the induction hypothesis for $i-1$ and 
defining $b_{i+1}=\pre b_i\pre b_i'\suf b_i'\suf b_i$ and $\tilde 
v_i=\pre v_i\pre v_i'\suf v_i'\suf v_i$ yields the computation for 
$T$ in the claim for $i$. 
Finally, we use $\mathcal{L}(A)=\mathcal{L}(T)$ to obtain from this 
new computation of $T$ the missing last computation of $A$. 
All mentioned computations of $A$ start with the same prefix 
\[\hat 
q_0\tr[8em]{\pre aa_ka_{k-1}\dots a_{i+1}}{u\pre vv_kv_{k-1}\dots 
v_{i+1}}{}\hat q_{i+1}\]
since $A$ is deterministic and all input 
and output in this prefix has remained unchanged. 
This concludes the induction step and thus the proof of 
\cref{clm:induction}. 
\end{proof} 

We now provide an illustrative example of an \fnft with unbounded 
trailing. 
According to \cref{thm:bt-necessary},
it has no equivalent \ftdfa{}, implying the strictness of the 
inclusion $\mathcal{L}(\ftdfa)\subseteq \mathcal{L}(\fnft)$; see 
\cref{fig:modelcompare}.

\begin{example}\label{ex:bt}
The transducer's alphabets are $\Sigma=\Gamma=\{0,1\}$ and its 
relation is %
\[\{\,(0^i10^j10,0^i)\mid 
i,j\in\N\,\}\,\cup\,\{\,(0^i10^j11,0^j)\mid 
i,j\in\N\,\}.\] 
This union contains exactly the pairs of the form 
$(0^i10^j1\beta,0^k)$, where $\beta$ is a bit valued $0$ or $1$ and 
$k=i$ if $\beta=0$ and $k=j$ if $\beta=1$. 
This relation is easily computed by an \fnft that 
nondeterministically guesses $\beta$, then 
copies either $0^i$ or $0^j$ to the output tape, and finally 
accepts or rejects depending on whether the guess was correct. 

No \tdfa can compute this relation, however, for the following 
intuitive reason. 
Consider an input-output pair $(0^i10^j1\beta,0^k)$ with 
sufficiently large $i$, $j$, and $k$. 
By the time the input head reaches the first $1$, the output head 
has either read most of $0^k$ already or has a large part of it 
still lying ahead. 
In the first case, the automaton may have potentially checked 
whether $k=i$,
but cannot remember how many zeroes of $0^k$ the output head has 
already passed
making it impossible to check whether $k=j$ if the input ends with 
$\beta=1$.
In the second case, the automaton can potentially still check 
whether $k=j$,
but cannot remember how many zeroes the input head has already 
passed, that is, the value of $i$. 
This makes it impossible to check whether $k=i$ if the input ends 
with $\beta=0$.
Therefore, the automaton fails on inputs with $\beta=1$ in the 
first case and on inputs with $\beta=0$ in the second case.

To show formally that the described transducer has unbounded 
trailing we use the variable names of \cref{def:bt}. 
For any given $\traillength\in\N$,
we can choose $u=w_1=w_2=\lambda$, where $\lambda$ denotes the 
empty word, $a=v=0^t$, $b_1=10^t10$, and $b_2=10^t11$. 
This yields two computations showing that the trailing bound must 
be at least $\traillength$, namely 
\[q_0\tr{\raisebox{-.3ex}{$0^t$}}{\lambda0^t}q_\indexcaseA
\tr[3em]{\raisebox{-.3ex}{$10^t10$}}{\lambda}f_\indexcaseA
\text{ and }
q_0\tr{\raisebox{-.3ex}{$0^t$}}{\lambda}q_\indexcaseB
\tr[3em]{\raisebox{-.3ex}{$10^t11$}}{0^t\lambda}f_\indexcaseB.\]
\end{example}

\section{Bounded Trailing is Undecidable}\label{sec:undecidability}
In this section, we prove that determining 
whether an \fnft has bounded trailing is undecidable.
This is achieved by reducing the halting problem 
on the empty input, which is known to be undecidable,
to the problem of determining whether an \fnft has bounded trailing.
We present a reduction via a third problem, 
namely determining whether a Turing machine
reaches infinitely many configurations on the empty input. 

We begin by formally defining the standard model of 
a deterministic Turing machine with a single tape
that is unbounded in both directions. 

\begin{definition}
A \emph{Turing machine} is a sextuple $M=(Q, \Gamma, \blank, q_0, 
F, \delta)$ consisting of 
\begin{itemize}
\item a finite, nonempty \emph{set of states} $Q$,
\item a finite, nonempty \emph{alphabet} $\Gamma$,
\item a blank symbol $\blank\not\in\Gamma$,
\item an initial state $q_0\in Q$,  
\item a set $F\subseteq Q$ of accepting states, and
\item a (partial) transition function\\$\delta : Q \times 
(\Gamma\cup\{\blank\})\rightarrow Q \times \Gamma \times 
\{\textnormal{Left}, \textnormal{Right}\}$.
\end{itemize}
\end{definition}

A \emph{configuration} of a Turing machine consists of 
its current state,
the content of the tape, 
and the position of the head on the tape.
We will only consider the computations of a Turing machine 
on the empty input, 
the initial configuration thus always consists of 
the initial state~$q_0$, 
a tape containing only blank symbols, 
and the head scanning one of them. 
A configuration is called \emph{accepting} 
if its current state $q$ is accepting, i.e., if $q\in F$.
A configuration is called \emph{halting} if it is 
accepting or the transition function
is undefined for the current state $q\in Q$ and 
the symbol $a\in\Gamma$ currently
scanned by the head. 
If a configuration is not halting, the next configuration 
reached in one step of the Turing machine's computation
is obtained by updating the current state, 
writing a non-blank symbol to the tape's cell
scanned by the head, and moving the head either one cell 
to the left or one cell to the right. 

Any configuration reached during the Turing machine's computation
on the empty input
consists of a finite contiguous sequence of non-blank symbols 
and the position of the head scanning
either any symbol within this sequence or 
one of the two blank symbols delimiting it. 
Hence, we can represent a configuration of this machine
as a finite sequence of cells of two types: 
Every configuration contains exactly one cell of the first type, 
namely the one currently scanned by the head. 
In our representation, this type of cell contains 
some potentially blank symbol and the current state. 
The second type contains a non-blank symbol only. 
The initial configuration is represented by a single cell 
containing the blank symbol and the initial state---recall that 
we consider only the computation on the empty input. 

A Turing machine halts on the empty input if it reaches 
a halting configuration during its computation
starting in the initial configuration. 
The undecidability of determining 
whether a halting configuration can be reached is 
well known~\cite{T37}. 
A straightforward reduction shows
that it is also undecidable whether a given Turing machine reaches
infinitely many different configurations 
during its computation on the empty input.

\begin{lemma}\label{lem:undecidable}
The problem of determining whether a Turing machine reaches 
infinitely many different configurations during its computation 
on the empty input is undecidable.
\end{lemma}

\begin{proof}
We prove the lemma by contradiction.
Suppose that there is an algorithm $A_\infty$ that decides for 
every given Turing machine $M$ whether it reaches infinitely many 
different configurations on the empty input. 
We will use $A_\infty$ to design an algorithm $A_{\text{H}}$ that 
decides for any Turing machine $M$ whether it reaches a halting 
configuration on the empty input. 
The latter problem is known to be undecidable, yielding the desired 
contradiction.

Given a Turing machine $M$, algorithm $A_{\text{H}}$ first invokes 
$A_\infty$ to decide whether $M$ reaches infinitely many different 
configurations on the empty input. 
If it does, then $A_{\text{H}}$ outputs ``No'' because reaching a 
halting configuration implies reaching only 
finitely many configurations in total. 
Otherwise, $A_{\text{H}}$ simulates $M$'s deterministic computation 
on the empty input step by step, remembering all configurations, 
until either a halting or a previously encountered configuration is 
reached.
Then $A_{\text{H}}$ outputs ``Yes'' in the former case and ``No'' 
in the latter. 
\end{proof}

Because both the set of states and the alphabet are finite, the set
of all configurations of a bounded length is necessarily finite.
It follows that a Turing machine $M$ reaches 
infinitely many configurations if and only if 
it reaches configurations of arbitrary length.

\begin{corollary}\label{cor:undecidable}
The problem of determining whether a Turing machine reaches 
configurations of arbitrary length on the empty input 
is undecidable. 
\end{corollary}

We now show how to construct from a given deterministic 
Turing machine $M$ an \fnft{} $T$ that has unbounded trailing
if and only if $M$ reaches configurations of arbitrary length.

A valid input for $T$ is a sequence of $M$'s configurations 
followed by one of two special symbols that we call 
\emph{mode indicators}. 
The configurations are represented 
by finite sequences of cells as described in the previous section
and separated from each other by a dedicated symbol 
not occurring anywhere else. 
The two mode indicators are represented by a cell containing either 
of the two words in $\{\text{copy},\text{step}\}$. 
The transducer $T$ starts its computation by nondeterministically 
guessing the mode indicator and then 
operates in the corresponding mode described below.
If the guess turns out to be wrong or if the input is 
invalid in any way, the computation is aborted and 
the transducer rejects the input. 
The two possible types of input-output pairs after an accepting 
computation are depicted in \cref{fig:comps}.

\medskip
\begin{description}
\item[\bf Copy Mode]\hfill\\
The input is output symbol by symbol, omitting 
the mode indicator in the end.
\item[\bf Step Mode]\hfill\\
In the first step, output $M$'s entire fixed 
initial configuration $c_0$ 
while reading and remembering the first input symbol. 
Then read from the sequence in the input one configuration $c$
after the other. While reading $c$, compute and output 
the successor configuration $c'$ that $M$ reaches from $c$ 
in one computation step. 
Of course, this all assumes that such a configuration $c'$ exists;
otherwise, $T$ aborts the computation 
as it does for invalid inputs. 
\end{description}
\medskip

To see that $T$ can in fact realize these computations, 
observe that the changes necessary to turn a configuration $c$ 
into its successor configuration $c'$ do, on the one hand, 
only depend on the single type-one cell of $c$ containing 
the currently scanned symbol and the current state and, 
on the other hand, 
only affect the immediate proximity of this cell. 

Hence, the successor configuration $c'$ can indeed be computed 
from $c$ by a finite transducer that essentially is still copying 
each configuration symbol by symbol, but letting the input head 
run ahead by one cell while keeping the contents of the three 
most recent input cells in a buffer. 
This allows the transducer to modify the configuration 
in the right place to produce the successor configuration 
even for configuration changes of the Turing machine that
move its only head to the left. 
Transducer $T$ can be effectively computed from any given Turing 
machine $M$; we omit the technicalities. 

\begin{figure}[ht]
\begin{subfigure}{\linewidth}
\centering
\begin{tabular}{l}
\begin{tabular}{*3{|C{1.6em}}|C{2.4em}*2{|C{1.6em}}|}
\cline{1-3}\cline{5-6}
$c_1$&$c_2$&$c_3$&$\cdots$&$c_k$&\!copy\!\\
\cline{1-3}\cline{5-6}
\end{tabular}\\[.4ex]
\begin{tabular}{*3{|C{1.6em}}|C{2.4em}*1{|C{1.6em}}|}
\cline{1-3}\cline{5-5}
$c_1$&$c_2$&$c_3$&$\cdots$&$c_k$\\
\cline{1-3}\cline{5-5}
\end{tabular}
\end{tabular}\\[1.2em]
\begin{tabular}{l}
\begin{tabular}{*3{|C{1.6em}}|C{2.4em}*2{|C{1.6em}}|}
\cline{1-3}\cline{5-6}
$c_1$&$c_2$&$c_3$&$\cdots$&$c_k$&step\!\\
\cline{1-3}\cline{5-6}
\end{tabular}\\[.4ex]
\begin{tabular}{*3{|C{1.6em}}|C{2.4em}*2{|C{1.6em}}|}
\cline{1-3}\cline{5-6}
$c_0$&$c_1'$&$c_2'$&$\cdots$&$c_{k-1}'\!\!$&$c_k'$\\
\cline{1-3}\cline{5-6}
\end{tabular}
\end{tabular}
\caption{The two types of valid inputs for $T$ and the 
corresponding accepted outputs. Each $c_i$ with $i>0$ encodes an 
arbitrary configuration of $M$, $c_0$ is its initial configuration, 
and $c_i'$ is $c_i$'s successor configuration 
under a single computation step of $M$.
}
\label{fig:comps}
\end{subfigure}\\[1em]
\begin{subfigure}{\linewidth}
\centering
\begin{tabular}{l}
\begin{tabular}{*3{|C{1.6em}}|C{2.4em}*2{|C{1.6em}}|}
\cline{1-3}\cline{5-6}
$c_0$&$c_1$&$c_2$&$\cdots$&$c_k$&\!copy\!\\
\cline{1-3}\cline{5-6}
\end{tabular}\\[.4ex]
\begin{tabular}{*3{|C{1.6em}}|C{2.4em}*1{|C{1.6em}}|}
\cline{1-3}\cline{5-5}
$c_0$&$c_1$&$c_2$&$\cdots$&$c_k$\\
\cline{1-3}\cline{5-5}
\end{tabular}
\end{tabular}\\[1.2em]
\begin{tabular}{l}
\begin{tabular}{*3{|C{1.6em}}|C{2.4em}*2{|C{1.6em}}|}
\cline{1-3}\cline{5-6}
$c_0$&$c_1$&$c_2$&$\cdots$&$c_k$&step\!\\
\cline{1-3}\cline{5-6}
\end{tabular}\\[.4ex]
\begin{tabular}{*3{|C{1.6em}}|C{2.4em}*2{|C{1.6em}}|}
\cline{1-3}\cline{5-6}
$c_0$&$c_1$&$c_2$&$\cdots$&$c_k$&$c_{k+1}\!\!$\\
\cline{1-3}\cline{5-6}
\end{tabular}
\end{tabular}
\caption{Two accepting computation patterns of $T$ that make 
unbounded trailing inevitable if $M$ reaches arbitrarily long 
configurations. 
Here, $c_0$ is still the initial configuration of $M$ on the empty 
input, but $c_{i+1}$ is now the successor configuration of $c_i$ 
under a single computation step of $M$.}
\label{fig:bt_inst}
\end{subfigure}
\caption{Accepting computation patterns of transducer $T$, which 
depends on Turing machine $M$.}
\label{fig:undecidability}
\end{figure}

The described transducer $T$ is \emph{unambiguous}; that is, 
there is at most one accepting computation for every input word. 
This is easily checked as follows. 
On the one hand, $T$ rejects all invalid inputs anyway. 
On a valid input, on the other hand, $T$ takes only a single
nondeterministic decision to choose the operating mode and then 
the computation proceeds deterministically, 
being accepting for only one of the two choices, 
depending on the mode indicator at the end of the input word.
Transducer $T$'s being unambiguous implies that the undecidability
stated in \cref{thm:undecidable} remains valid even when we are 
restricted to unambiguous \nft{}s instead of \fnft{}s. 
Namely, unambiguity implies functionality and thus
finite-valuedness---but of course not vice versa---which trivially 
means that \cref{thm:bt-sufficient,thm:bt-necessary} hold for 
 functional and unambiguous \nft{}s too. 
Nevertheless, unambiguous transducers are just as expressive 
as functional ones~\cite{CG99}. 
We state $T$'s relevant properties
formally. 
\begin{claim}\label{clm:red-func}
$T$ is unambiguous and thus also functional and finite-valued; that 
is, $T$ is an \funnft{} and an \fnft{}.
\end{claim}

We will now sketch the proof of the crucial connection between 
the length of $M$'s configurations and $T$'s trailing. 

\begin{claim}\label{clm:bt-undec}
 $T$ has unbounded trailing
if and only if Turing machine $M$ reaches configurations of 
arbitrary length during its computation on the empty input.
\end{claim}

\begin{proof}
The transducer $T$ takes only one nondeterministic decision 
during any computation, namely to operate in either copy or 
step mode, the rest of the computation is deterministic. 
According to \cref{def:bt}, the only way for any trailing 
to occur is therefore a pair of two computations, 
one in copy mode and one in step mode, 
producing consistent output words. 
One of these two output words must be a prefix of the other; 
we call it the common prefix of the two computations. %
See \cref{fig:bt_inst} for an example of the arising situation.

In the following paragraph we argue why the configurations 
within the common prefix represent a valid computation of $M$ and 
why the current trail length is equal to the length of the 
configuration currently read, up to a constant. 

Since transducer $T$ always starts by outputting 
the initial configuration $c_0$ in step mode, t
his has to be the first configuration on the output tape 
in copy mode as well. 
In copy mode, $T$ can only write this configuration $c_0$ 
to the beginning of the output tape 
if $c_0$ is also the first configuration on the input tape. 
The step mode behavior now ensures that the second configuration on 
the output tape is the successor of $c_0$, we call it $c_1$. 
Iterating this argument, we see that the common prefix 
indeed contains a valid computation $c_0,c_1,\dots,c_k$, 
where each configuration is the successor of the previous one 
under one computation step of $M$.
Finally, the trail is always as long as the configuration 
that is currently being read, up to the size of the one cell 
by which the input head is leading. 
Thus the current trail length is indeed always, up to constant, 
equal to the length of the current configuration in 
the simulation of the computation of $M$ on the empty input word. 
This concludes the proof of \cref{clm:bt-undec}.
\end{proof}

Finally, we obtain the main result of this section by combining 
\cref{cor:undecidable,clm:red-func,clm:bt-undec}. 

\begin{theorem}\label{thm:undecidable}
Whether a given \fnft{} has bounded trailing is 
an undecidable problem.
\end{theorem}

\begin{proof}
We prove the theorem by contradiction.
Suppose that there is an algorithm $A$ deciding whether an \fnft{} 
has bounded trailing.
We construct an algorithm $A_\infty$ deciding whether a Turing 
machine $M$ reaches configurations of arbitrary length.
The latter problem is undecidable by \cref{cor:undecidable}, 
yielding a contradiction.

Given a Turing machine $M$, algorithm $A_\infty$ constructs the 
\fnft{} $T$
and uses algorithm $A$ to decide whether $T$ has bounded trailing.
If it does, then $A_\infty$ outputs ``No,'' otherwise it outputs 
``Yes.''
The correctness of algorithm $A_\infty$ follows from 
\cref{clm:bt-undec}.
\end{proof}

\section{Conclusion}
We have introduced the notion of bounded trailing, characterized 
\fnft{}s having 
an equivalent \ftdfa{} as those with bounded trailing, proved that 
it is undecidable whether an \fnft has bounded trailing---thereby 
proving \fnft{}s 
to be more powerful than \ftdfa{}s---and showed 
how to construct an equivalent \ftdfa given a 
trailing bound. 
Furthermore, all of these results hold true for functional and 
unambiguous transducers as well because 
\cref{thm:bt-sufficient,thm:bt-necessary} prove the 
characterization for any \fnft---including functional and 
unambiguous finite transducers---and because we have shown our 
criterion's undecidability using an unambiguous finite 
transducer, which is the most restrictive. 

On the one hand, our results mirror Choffrut's classical 
characterization of 
determinizable \nft{}s as those with the twinning 
property~\cite{CH79}. 
One the other hand, the undecidability of our bounded-trailing 
criterion contrasts with the twinning property being decidable in 
polynomial 
time~\cite{WK95}. 
Despite this undecidability, we provide a 
uniform construction that transforms an \fnft{} into an \ftdfa{} 
whenever possible. 

Fischer and Rosenberg~\cite{FR68} have shown already more than half 
a century ago 
that deciding whether a $\textrm{2t-}$\nfa 
(or equivalently, a \lnft) 
can be transformed 
into an equivalent \tdfa is an undecidable problem; 
however, their proof requires relations that are not finite-valued 
and thus does not yield any result for $\fnft$s. 
Moreover, neither a characterization of transducers with an 
equivalent deterministic automaton nor a method to construct 
an equivalent deterministic automaton if it exists 
has been described so far.  
The present paper addresses all of these issues for the case of 
\fnft{}s.

We emphasize again that the bounded trailing property defined in
this paper does not characterize \nft{}s having an equivalent \tdfa
unless the relation is finite-valued; 
a corresponding counterexample of a transducer with both an 
equivalent \tdfa{} and unbounded trailing is found in 
\cref{ex:counterexample}. 
Finding a reasonable characterization of \nft{}s 
with an equivalent \tdfa{} and a method to transform the former 
into the latter if possible remains as an open research opportunity.

\section*{Acknowledgments}
We are grateful to Juraj Hromkovi\v{c} and Richard Kr\'{a}lovi\v{c} 
for inspiring discussions, and 
thank the anonymous reviewers for their careful 
reading of this paper and suggestions. 

This research is supported by the Swiss National Science Foundation 
grant ``Big Data Monitoring''(167162). 

The authors are listed in alphabetical order.
\newpage

\bibliographystyle{plainurl}
\bibliography{bibliography}

\end{document}